\documentclass[a4paper,11pt]{article}
\usepackage{graphicx}
\usepackage[dvips]{psfrag}
\usepackage{url}
\usepackage{lscape}
\usepackage{booktabs}
\usepackage{amsthm}
\usepackage{amsmath}
\usepackage{bm}
\usepackage{mathtools}
\usepackage[norelsize,ruled,vlined]{algorithm2e}

\newcommand{\argmin}{\operatornamewithlimits{argmin}}

\makeatletter
\def\section{\@startsection {section}{1}{\z@}{-3.5ex plus -1ex minus -.2ex}{2.3ex plus .2ex}{\large\bf}}

\def\subsection{\@startsection {subsection}{2}{\z@}{-3.5ex plus -1ex minus 
-.2ex}{2.3ex plus .2ex}{\normalsize\bf}}

\def\subsubsection{\@startsection {subsubsection}{3}{\z@}{-3.5ex plus -1ex minus 
-.2ex}{2.3ex plus .2ex}{\normalsize\bf}}
\makeatother

\makeatletter
  \newcounter{subeqncnt}
  \def\thesubeqncnt{\alph{subeqncnt}}%%% ?u????”??†?v???`??
  \def\subequations{\begingroup%
     \stepcounter{equation}\edef\@tempa{\theequation}%
     \let\c@equation\c@subeqncnt\c@subeqncnt\z@ 
     \edef\theequation{\@tempa\noexpand\thesubeqncnt}}
  
  \makeatother

\makeatletter
\newcount\@cla
\newcount\@clb
\makeatother

\newtheorem{theorem}{Theorem}

\begin{document}

\title{Escherization with Generalized Distance Functions Focusing on Local Structural Similarity}

\author{
Yuichi Nagata\thanks{Corresponding author: nagata@is.tokushima-u.ac.jp}
\and Shinji Imahori\textsuperscript{\dag}
}

\date{}
\maketitle

{
\center 
%\vspace{-10mm}
\footnotesize
\textsuperscript{*} Graduate School of Technology, Industrial and Social Sciences, Tokushima University, 2-1 Minami-jousanjima, Tokushima-shi, Tokushima 770-8506, Japan \\
\textsuperscript{\textsuperscript{\dag}} Department of Information and System Engineering Faculty of Science and Engineering, Chuo University, Bunkyo-ku, Tokyo 112-8551, Japan \\
}
\maketitle

\maketitle

\begin{abstract}
The Escherization problem involves finding a closed figure that tiles the plane that is most similar to a given goal figure. In Koizumi and Sugihara's formulation of the Escherization problem, the tile and goal figures are represented as $n$-point polygons where the similarity between them is measured based on the difference in the positions between the corresponding points. This paper presents alternative similarity measures (distance functions) suitable for this problem. The proposed distance functions focus on the similarity of local structures in several different manners. The designed distance functions are incorporated into a recently developed framework of the exhaustive search of the templates for the Escherization problem. Efficient exhaustive and incomplete search algorithms for the formulated problems are also developed to obtain results within a reasonable computation time. Experimental results showed that the proposed algorithms found satisfactory tile shapes for fairly complicated goal figures in a reasonable computation time. 

\bigskip


\end{abstract}

\section{Introduction} \label{sec:intro}

A tiling is a collection of shapes, called tiles, which cover the plane with no gaps and no overlaps. Tiling has attracted considerable attention owing to its practical and mathematical aspects. The Dutch artist M.~C.~Escher studied tiling from a mathematical perspective and created many artistic tilings, each of which consists of one or a few recognizable figures such as animals. Such artistic tilings are now called Escher tilings. Designing an artistic Escher-like tiling is a highly intellectual task because it is difficult to create meaningful tile figures while satisfying the constraints to enable tiling. 

To generate Escher-like tilings automatically, Kaplan and Salesin \cite{kaplan2000escherization} formulated the following optimization problem.  
\smallskip

\noindent
{\it Escherization problem}: Given a closed plane figure $S$ (goal figure), find a closed figure $T$ such that 
\begin{enumerate}
\setlength{\parskip}{0cm} % ???
\setlength{\itemsep}{1mm} % ???
\item
$T$ is as close as possible to $S$, and  
\item
copies of $T$ fit together to form a tiling of the plane. 
\end{enumerate}
To tackle this problem, they introduced parameterizations of possible tile shapes for the 93 types of isohedral tilings, which are a class of tiling that is sufficiently flexible to express tiling patterns as well as mathematically tractable. For each isohedral type, the nature of tile shapes can be expressed by a template, a polygon composed of at most six vertices, from which all possible tile shapes are obtained by moving the vertices and deforming the edges under the constraints specified by the template. To measure the similarity between the two shapes $S$ and $T$, they employed a metric developed by Arkin et al. \cite{arkin1991efficiently}, which is based on the $L_2$ distance of their turn angle representations. They developed a simulated annealing (SA) algorithm to optimize the formulated problem and successfully found satisfactory tile shapes for convex or nearly convex goal figures. However, their SA algorithm did not perform well for non-convex goal figures.

%In their parameterizations, possible tile shapes were represented by the template for each isohedral type and the edges of the template were represented by subdivision curves. 

Koizumi and Sugihara \cite{koizumi2011maximum} formulated the Escherization problem so that the formulated problem has a closed-form solution. In their formulation, the tile and goal shapes were represented as $n$-point polygons. For each isohedral type, possible tile shapes were parameterized by the template with $n$ points on the edges; possible positions of the $n$ points are obtained by moving the $n$ points of the template under the specified constraints. This parameterization made it possible to express the constraint conditions as a linear system. The Procrustes distance \cite{werman1995similarity} was employed to measure the similarity between the two shapes (polygons) $S$ and $T$. The Procrustes distance is rotation and scale-invariant distance metric, which measures the difference in the positions of the $n$ points between the two polygons. Under this formulation, the Escherization problem was reduced to an eigenvalue problem. Koizumi and Sugihara's method performed well for both convex and non-convex goal figures. 

Several enhancements of Koizumi and Sugihara's formulation have been developed. Imahori and Sakai \cite{imahori2013escher} parameterized tile shapes by assigning a different number of points to each edge of the template for each isohedral type, whereas all edges were assigned the same number of points in Koizumi and Sugihara's original formulation. This extension provides a considerable flexibility in the possible tile shapes and improves the quality of the obtained tile shapes. However, a huge number of different eigenvalue problems must be solved when all possible configurations of the templates are considered, and the exhaustive search of the templates was computationally impractical. As a compromise, they proposed a local search algorithm to search only promising configurations of the templates. Recently, Nagata and Imahori \cite{nagata2019} developed an efficient algorithm to perform the exhaustive search of the templates, and it can now be performed in a reasonable computation time (e.g. 0.55 s and 9.01 s for 60- and 120-point goal polygons, respectively). 

In the Escherization problem, it is also important to employ an appropriate similarity measure (or distance function) to obtain satisfactory tile shapes. In our preliminary study \cite{imahori2015escher}, we introduced weights to the Procrustes distance to emphasize the similarity with important parts of the goal polygon. In another preliminary study \cite{nagata2018escherization}, we proposed a distance function, which focuses on the similarity of local structures between the two polygons. The effectiveness of these distance functions has been demonstrated in each paper. However, these distance functions have not been incorporated into the exhaustive search of the templates because the required computational efforts were impractical. 

In this paper, we develop new distance functions suitable for Koizumi and Sugihara's formulation of the Escherization problem, aiming at better evaluating the similarity between the tile and goal shapes (polygons). Some of the distance functions proposed in this paper were presented in our preliminary studies \cite{imahori2015escher,nagata2018escherization}, but we further improve these distance functions by evaluating the similarity of local structures in more appropriate ways. Moreover, we develop efficient algorithms to perform the exhaustive search of the templates combined with the proposed distance functions. The main contributions of this paper are summarized as follows: (i) Several distance functions suitable for Koizumi and Sugihara's formulation of the Escherization problem are proposed. (ii) The proposed distance functions are incorporated into the exhaustive search of the templates and efficient exhaustive search algorithms are developed. (iii) The proposed distance functions and the Procrustes (or Euclidean) distance are discussed from a unified viewpoint in terms of the quality of the obtained tile shapes and the time complexity. 

The remainder of this paper is organized as follows. In Section \ref{sec:2}, Koizumi and Sugihara's formulation for the Escherization problem is described along with some related work and concepts. In Section \ref{sec:3}, several distance functions suitable for the Escherization problem are presented. In Section \ref{sec:4}, efficient algorithms for the exhaustive search of the templates combined with the proposed distance functions are presented. Experimental results are presented in Section \ref{sec:5}. The proposed distance functions are discussed in Section \ref{sec:6} and this paper is concluded in Section \ref{sec:7}.

\section{Background} \label{sec:2}

In this section, Koizumi and Sugihara's formulation of the Escherization problem is described along with some related work and concepts.

\subsection{Isohedral tilings}  \label{sec:2_Isohedral}

A {\it monohedral} tiling is one in which all the tiles have the same shape. If a monohedral tiling has a repeating structure where a parallel shift of the tiling can be mapped onto itself, this tiling is called an {\it isohedral} tiling. There are 93 types of isohedral tilings \cite{grunbaum1987tilings}, which are referred to individually as IH1, IH2, \dots, IH93. Fig. \ref{fig:tiling_IH47} illustrates an example of an isohedral tiling belonging to IH47 with a few technical terms. For an isohedral tiling, a point that is shared by the boundaries of three or more tiles is called a {\it tiling vertex}, whereas a surface that is shared by the boundaries of two tiles is called a {\it tiling edge}. A tiling polygon is defined as the polygon formed by connecting the tiling vertices of a tile.

\begin{figure} [t] %%%%%%%%%%%%%%%%%%%%%%%%%%%%%
\centering
\includegraphics[scale=0.45,keepaspectratio,clip]{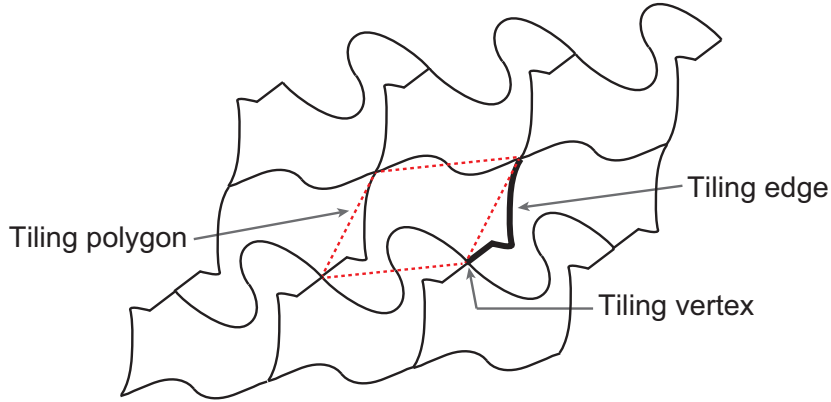}
\caption{Example of an isohedral tiling (IH47)}
\label{fig:tiling_IH47}
\end{figure} %%%%%%%%%%%%%%%%%%%%%%%%%%%%%

For each IH type, the nature of the possible tile shapes can be represented by a template \cite{kaplan2009introductory}. The template represents a tiling polygon, from which all possible tile shapes are obtained by moving the tiling vertices and deforming the tiling edges under the constraints specified by the template. Fig. \ref{fig:template}(a) illustrates a template of IH47; this template indicates that the tiling polygon is any quadrilateral consisting of two opposite {\tt J} edges parallel to each other and two arbitrary {\tt S} edges. There are four types of tiling edges (types {\tt J}, {\tt S}, {\tt U}, and {\tt I}) depending on how they can be deformed. It is intuitive and convenient to represent the edge types with colored arrowheads in the template (see Tom McLean's website: \url{https://www.jaapsch.net/tilings/mclean/index.html}), where {\tt J} and {\tt S} edges are shown in the figure. Possible deformations for the four tiling edge types are as follows:
\begin{itemize}
\setlength{\itemindent}{15pt}   %5. 最初のインデント
\item[{\tt J} edge:] 
This edge can be deformed into any shape, but the corresponding {\tt J} edge must also be deformed into the same shape, as suggested by the arrowheads. 
\item[{\tt S} edge:] 
This edge must be symmetric with respect to its midpoint. 
\item[{\tt U} edge:] 
This edge must be symmetric with respect to a line through its midpoint and orthogonal to it. 
\item[{\tt I} edge:] 
This edge must be a straight line. 
\end{itemize}
Each template also specifies how the tiles are arranged, where the tiles must be placed such that the arrowheads with the same color and shape coincide with each other as shown in Fig. \ref{fig:template}(b).

\begin{figure} [t] %%%%%%%%%%%%%%%%%%%%%%%%%%%%%
\centering
\includegraphics[scale=0.45,keepaspectratio,clip]{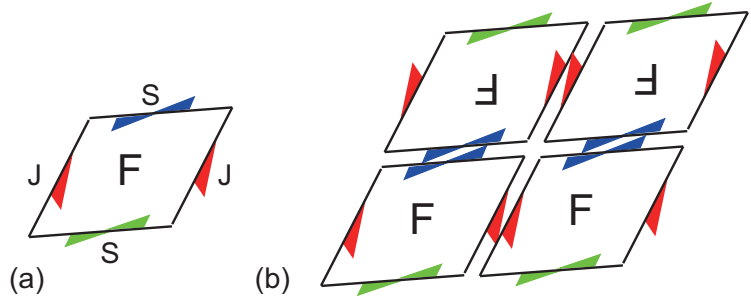}
\caption{(a) Template of IH47 and (b) the adjacency relationship between the tiles}
\label{fig:template}
\end{figure} %%%%%%%%%%%%%%%%%%%%%%%%%%%%%

The 93 isohedral types can be classified into the nine most general isohedral types \cite{Schattschneider2004} shown in Fig. \ref{fig:template_all}. Any isohedral type can be obtained from one of these nine isohedral types by removing one or more tiling edges and/or by replacing {\tt J} and {\tt S} edges with more symmetrical edge types ({\tt I} and {\tt U} edges). 

\begin{figure} [t] %%%%%%%%%%%%%%%%%%%%%%%%%%%%
\centering
\includegraphics[scale=0.50,keepaspectratio,clip]{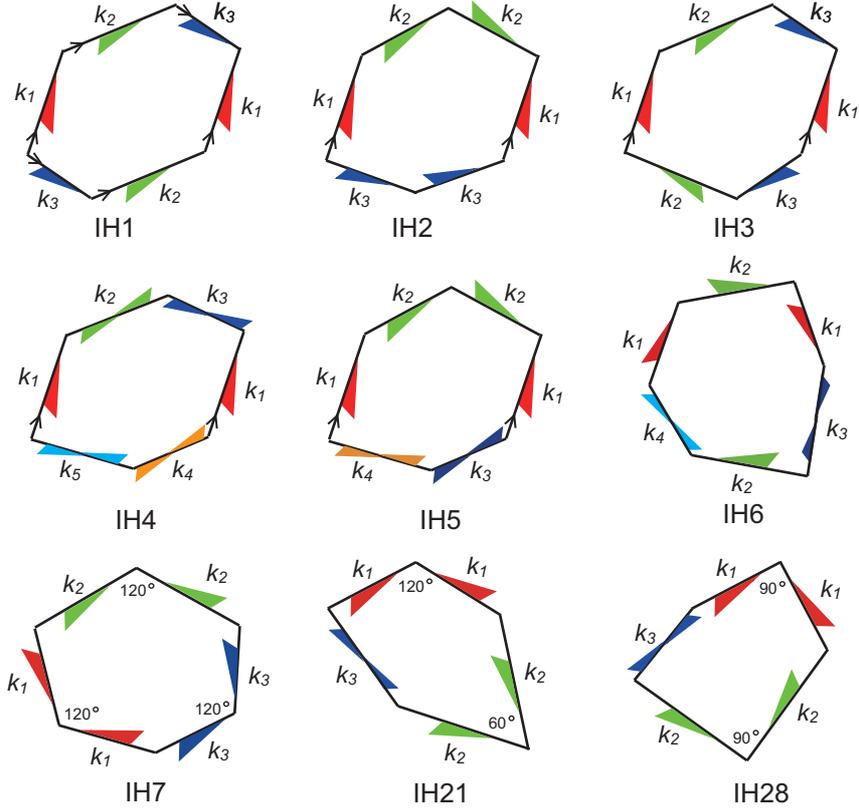}
\caption{Templates of the nine most general isohedral types. Two opposite {\tt J} edges marked with $\wedge$ are parallel to each other. Consecutive {\tt J} edges must form a specified angle if designated. Otherwise, each of the pairs of J edges is glide-reflection symmetric with respect to the $x$-axis or $y$-axis.} 
\label{fig:template_all}
\end{figure} %%%%%%%%%%%%%%%%%%%%%%%%%%%%%

\subsection{Koizumi and Sugihara's formulation and its extensions}   \label{sec:2_parameterization}

Koizumi and Sugihara \cite{koizumi2011maximum} modeled the tile shape as an $n$-point polygon. The tile shape (polygon) $U$ is constrained to form an isohedral tile and the constraint conditions depend on the isohedral type. For example, the template of IH47 is represented as shown in Fig.\ref{fig:template_point_IH47}, where exactly one point must be placed at each of the tiling vertices (black circles) and the remaining points are placed on the tiling edges (white circles). We denote the numbers of points placed on the tiling edges as $k_1, k_2, \dots$ as illustrated in the figure. This template represents the possible positions of the $n$ points for IH47; for example, the $n$ points can be moved as illustrated in the figure. In Koizumi and Sugihara's original formulation, the same number of points were assigned to every tiling edge, i.e., $k_1 = k_2 = \dots \ $. Subsequently, Imahori and Sakai \cite{imahori2013escher} extended this formulation by assigning different numbers of points to the tiling edges, which allows considerable flexibility in the possible tile shapes. 

\begin{figure} [t] %%%%%%%%%%%%%%%%%%%%%%%%%%%%%
\centering
\includegraphics[scale=0.60,keepaspectratio,clip]{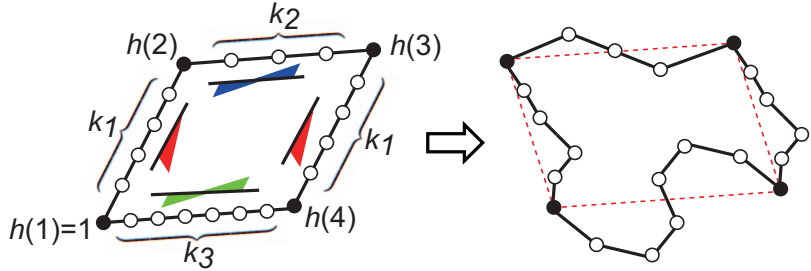}
\caption{Template of IH47 for a specific assignment of the points to the tiling edges (left), and an example of a possible tile shape (right) }
\label{fig:template_point_IH47}
\end{figure} %%%%%%%%%%%%%%%%%%%%%%%%%%%%%

The $n$ points on a template are indexed clockwise by $1, 2, \dots, n$, starting from one of the tiling vertices. Let $\bm{\hat{u}}_i = (x_i, y_i)^{\top}$ be the coordinates of the $i$th point of the tile shape $U$. The tile shape $U$ is then represented as a $2n$-dimensional vector $\bm{u} = {(x_1, x_2, \dots, x_n, y_1, y_2, \dots, y_n)}^{\top}$. For each of the isohedral types, constraint conditions imposed on the possible tile shapes have linearity. For example, from the template of IH47, the coordinates of the $n$ points must satisfy the following equation:
\begin{equation}
\label{eq:A_IH47}
\left\{
\begin{array}{llll}
\! \bm{ \hat{u}}_{h(1)+i} - \bm{ \hat{u}}_{h(1)} & \!\!=\!\! & \bm{ \hat{u}}_{h(4)-i} - \bm{ \hat{u}}_{h(4)}     & (i=1, \dots, k_1+1) \\
\! \bm{ \hat{u}}_{h(2)+i} - \bm{ \hat{u}}_{h(2)} & \!\!=\!\! & -(\bm{ \hat{u}}_{h(3)-i} - \bm{ \hat{u}}_{h(3)})  & (i=1, \dots, \lfloor \frac{k_2+1}{2} \rfloor) \\
\! \bm{ \hat{u}}_{h(4)+i} - \bm{ \hat{u}}_{h(4)} & \!\!=\!\! & -(\bm{ \hat{u}}_{h(5)-i} - \bm{ \hat{u}}_{h(1)})  & (i=1, \dots, \lfloor \frac{k_3+1}{2} \rfloor)
\end{array}%
\right.,
\end{equation}
where $h(s) \ (s =1, \dots, 4)$ is the index of the $s$th tiling vertex and $h(5)$ is defined as $n+1$. 

Eq.~\ref{eq:A_IH47} is a homogeneous system of linear equations, and let $\bm{b_1}, \bm{b_2}, \dots, \bm{b_m}$ be a set of $m$ linearly independent solutions of this equation ($m$ is the degree of freedom of this system). A general solution to Eq.~\ref{eq:A_IH47} is then given by
\begin{equation}
\label{eq:u=Bxi}
\bm{u} = \xi_1 \bm{b_1} + \xi_2 \bm{b_2} + \cdots + \xi_m \bm{b_m} = B \bm{\xi},
\end{equation}
where $B=(\bm{b_1} \ \bm{b_2} \ \dots \ \bm{b_m})$ is a $2n \times m$ matrix and $\bm{\xi} = (\xi_1, \xi_2, \dots, \xi_m)^{\top}$ is a parameter vector. 
Here, we assume that the vectors $\bm{b_1}, \bm{b_2}, \dots, \bm{b_m}$ are mutually orthonormal. In fact, for every isohedral type, the tile shape $U$ can be parameterized in the form of Eq.~\ref{eq:u=Bxi}, where the matrix $B$ depends on the isohedral type and the assignment of the $n$ points to the tiling edges. 

The goal figure is also represented as an $n$-point polygon $W$ and the $n$ points are indexed clockwise by $1, 2, \dots, n$. Let $\bm{\hat{w}}_i = (x^w_i, y^w_i)^{\top}$ be the coordinates of the $i$th point of $W$. The goal shape $W$ is then represented as a $2n$-dimensional vector $\bm{w} = {(x^w_1, x^w_2, \dots, x^w_n, y^w_1, y^w_2, \dots, y^w_n)}^{\top}$. 

Koizumi and Sugihara employed the Procrustes distance \cite{werman1995similarity} to measure the similarity between the two polygons $U$ and $W$. Let the terms $U$ and $W$ be also used to denote $2 \times n$ matrices defined by 
\begin{equation}
\label{eq:U_W}	
U = \left(
    \begin{array}{cccc}
      x_1 & x_2 & \dots & x_n \\
      y_1 & y_2 & \dots & y_n 
    \end{array}
  \right)
, \ \
W = \left(
    \begin{array}{cccc}
      x^w_1 & x^w_2 & \dots & x^w_n \\
      y^w_1 & y^w_2 & \dots & y^w_n 
    \end{array}
  \right).
\end{equation}
Let $\|X\|$ be the Frobenius norm of a matrix $X$. The Procrustes distance $d_P(U,W)$ is defined as follows: 
\begin{equation}
\label{eq:procrustes}
d_P^2(U,W) = \min_{s, \theta} {\left\| sR(\theta) U - \frac{W}{\|W\|}  \right\|}^2, 
%= 1 - \frac{ {\|UW^{\top}\|}^2 + 2\det(U{W}^{\top}) }{{\|U\|}^2 {\|W\|}^2},  
\end{equation}
where $s$ is a scalar expressing expansion (or contraction) and $R(\theta)$ is the rotation matrix by angle $\theta$. From the definition, the Procrustes distance is scale and rotation-invariant\footnote{According to the original definition, $U$ and $W$ are assumed to have the same centroid and this turns out to be the optimal translation of $U$ for minimizing the distance value.}. The property of rotation-invariance is indispensable\footnote{The property of scale-invariance is not necessary because the parameterized tile shape $U$ can be of any size. Therefore, a distance measure defined as $\min_{\theta} {\left\| R(\theta) U - W  \right\|}^2$ yields the same result.} for some isohedral types. For example, the two consecutive {\tt J} edges (specified by the green arrows) in the template of IH5 (see Fig. \ref{fig:template_all}) must be parameterized such that they make equal and opposite angles with the $x$-axis (or $y$-axis). Therefore, the parameterized tile shape $U$ can only appear in a specific orientation. This assumption is required to express the constraint conditions as linear equations for IH2, IH3, IH5, and IH6. For other isohedral types (IH1, IH4, IH7, IH21, and IH28), the same result is obtained (after the optimization procedure described later) except for the size if a more simple distance measure defined as 
\begin{equation}
\label{eq:euclid}
d_E^2(U,W) = {\left\| U - W \right\|}^2 = {\left\| \bm{u}-\bm{w} \right\|}^2
= \sum_{i=1}^n {\left\| \bm{\hat{u}_i}-\bm{\hat{w}_i} \right\|}^2 
\end{equation}
is used because the tile shapes are parameterized such that they can appear in any orientation \cite{nagata2019}. We refer to $d_E(U,W)$ as the Euclidean distance. 

When the Euclidean distance can be used, from Eqs.~\ref{eq:u=Bxi} and \ref{eq:euclid}, the Escherization problem is formulated as the following unconstrained optimization problem:
\begin{equation}
\label{eq:formulation_Euclid}
\argmin_{\bm{\xi}} \ {\left\| B \bm{\xi}-\bm{w} \right\|}^2. 
\end{equation}
This is a least-squares problem and the solution is given by 
$\bm{\xi^*} = {(B^{\top} B)}^{-1} B^{\top} \bm{w} = B^{\top} \bm{w}$ with the minimum value $-{\bm{\xi^*}}^{\top} \bm{\xi^*} + {\bm{w}}^{\top} \bm{w}$. The coordinates of the optimal tile shape $\bm{u^*}$ are obtained by $\bm{u^*} = B \bm{\xi^*}$. Note that the optimal tile shape occasionally has self-intersection(s) because the constraint conditions expressed as Eq.~\ref{eq:u=Bxi} do not exclude this case. Therefore, such a tile shape must be discarded. 

When calculating the Euclidean (or Procrustes) distance between the two polygons $U$ and $W$, we need to consider $n$ different numbering schemes for the goal polygon $W$. We denote by $W_j \ (j=1, 2, \dots, n)$ the goal polygon that is renumbered starting from the $j$th point in the original numbering. The coordinates of $W_j$ are represented as a $2n$-dimensional vector $\bm{w_j}$ in the same way as $\bm{w}$. Sometimes we will denote $W_j$ ($\bm{w_j}$) simply as $W$ ($\bm{w}$) for simplicity. 

Let $I$ be a set of indices for the isohedral types and $K_i \ (i \in I)$ a set of all possible configurations for the assignment of the $n$ points to the tiling edges of the template for an isohedral type IH$i$. For example, from the template of IH47 (see Fig. \ref{fig:template_point_IH47}), $K_{47}=\{(k_1, k_2) \mid 0 \leq k_1, 0 \leq k_2, 2 k_1 + k_2 \leq n-4 \}$, where $k_3$ is determined by $k_3=n-4-(2 k_1+k_2)$. As the matrix $B$ depends on $i \in I$ and $k \in K_i$, we denote it as $B_{ik}$. Let $J = \{1, 2, \dots, n\}$ be a set of the indices of the start point for the $n$ different numbering schemes of the goal polygon $W$. If we attempt to perform an exhaustive search, we need to solve the optimization problem 
\begin{equation}
\label{eq:formulation_ikj_Euclid}
\argmin_{\bm{\xi}} \ {\left\| B_{ik} \bm{\xi}-\bm{w_j} \right\|}^2, 
\end{equation}
for all possible combinations of $i \in I$, $k \in K_i$, and $j \in J$, and then select the best solution. For each triplet $(i,k,j)$, the optimal value is computed by
\begin{equation}
\label{eq:eval_ikj_Euclid}
eval_{ikj}^E = -{\bm{\xi^*_{ikj}}}^{\top} \bm{\xi^*_{ikj}} + {\bm{w}}^{\top} \bm{w}, 
\end{equation}
where $\bm{\xi^*_{ikj}} = {B_{ik}}^{\top} \bm{w_j}$.
We refer to this as the {\it exhaustive search of the templates}. Note that for some isohedral types (IH2, IH3, IH5, and IH6), we need to use the Procrustes distance and the optimization problem (\ref{eq:formulation_ikj_Euclid}) is modified into an eigenvalue problem. The case where the Procrustes distance must be used is described in Appendix A.

The exhaustive search of the templates was very time-consuming because (i) the order of $K_i$ is $O(n^3)$ for IH5 and IH6 and $O(n^4)$ for IH4 (see Fig. \ref{fig:template_all}), and (ii) for each combination of $i$ and $k$, it took $O(n^3)$ time for computing ${eval}_{ikj}^E$ for all $j \in J$ \cite{imahori2013escher,imahori2015escher}. As a compromise, Imahori and Sakai \cite{imahori2013escher} proposed a local search algorithm to search only promising configurations of $(k,j) \in (K_i, J)$ for each isohedral type $i \in I$. Recently, Nagata and Imahori \cite{nagata2019} developed an efficient exhaustive search algorithm for the nine most general isohedral types (the outline is described in Section \ref{sec:4_Euclidean}), where it takes $O(n^2)$ time for computing ${eval}_{ikj}^E$ for all $j \in J$ and unnecessary computation of $eval^E_{ikj}$ are avoided. The exhaustive search of the templates can now be performed in a reasonable computation time.

\section{Escherization problem with the proposed distance functions} \label{sec:3}

This section proposes several distance functions suitable for the Escherization problem and these are incorporated into the exhaustive search of the templates. As stated in Section \ref{sec:2_parameterization}, we can use the Procrustes distance (Eq.~\ref{eq:procrustes}) for all isohedral types, but the Euclidean distance (Eq.~\ref{eq:euclid}) can be used for some isohedral types without changing the result. In this section, the proposed distance functions are described based on the case where the Euclidean distance is available because it is better for intuitive understanding and the essence of the mathematical structure does not change in the case of the Procrustes distance. Appendix A describes how the proposed distance functions and the formulated problems are modified when the Procrustes distance must be used.

\subsection{A general case}  \label{sec:3_general}

We consider a squared distance function $d_G^2(U,W)$ given by the quadratic form of the vector ${\bm u} - {\bm w}$ associated with a matrix $G$ defined as follows:
\begin{equation}
\label{eq:mahalanobis}
d_G^2(U,W) = ({\bm u} - {\bm w})^{\top} G ({\bm u} - {\bm w}).
\end{equation}
If $G$ is an identity matrix, $d_G^2(U,W)$ is equivalent to the squared Euclidean distance $d_E^2(U,W)$. The matrix $G$ can be any symmetric positive semi-definite matrix (of size $2n$) to fulfill the axiom of the distance (non-negativity, identity of indiscernibles, symmetry, and triangle inequality). As described later, it is possible to incorporate this distance into the exhaustive search of the templates. As alternative distance metrics expected to be suitable for the Escherization problem, we design several distance functions by specifying the value of $G$. 

\subsection{Weighted Euclidean distance} \label{sec:3_WE}

In our preliminary study \cite{imahori2015escher}, we introduced weights to the Procrustes distance to emphasize the similarity with important parts of goal polygons selected by the user. Here, we explain this concept in the Euclidean distance case (see Appendix A for the Procrustes distance case). 

Let $k_i$ be a positive weight assigned to the $i$th point of the goal polygon $W$. The weighted Euclidean (WE) distance $d_{WE}(U,W)$ is then defined by 
\begin{equation}
\label{eq:euclid_weight}
d_{WE}^2(U,W) = \sum_{i=1}^n k_i { \left\|{\bm{\hat{u}}_i} - {\bm{\hat{w}}_i}  \right\| }^2. 
%           = ({\bm u} - {\bm w})^{\top} G ({\bm u} - {\bm w}), 
\end{equation}
The WE distance can be regarded as a special case of $d_G^2(U,W)$, where the matrix $G$ is a $2n \times 2n$ diagonal matrix whose diagonal elements are given by $k_1, \dots, k_n, k_1, \dots, k_n$. Note that the matrix $G$ depends on $j \ \in J$ because the indices of the weighted points must be shifted depending on the start point of the numbering of the goal polygon $W$. Therefore, we need to define $G_j$ individually for the goal polygons $W_j \ (j=1, 2, \dots, n)$. 

\subsection{Adjacent difference (AD) distance} \label{sec:3_AD}

As expected from the definition of the (weighted) Euclidean distance (Eqs.~\ref{eq:euclid} and \ref{eq:euclid_weight}), the positions of the $n$ points of the tile polygon $U$ must be close to the corresponding positions of the goal polygon $W$ to reduce the distance value. On the other hand, it is also natural to measure the similarity between the two polygons based on the ``relative position'' of adjacent points rather than based on the ``absolute position''.  For example, this type of similarity measure is commonly used in the research field of the 3D surface modeling and its editing \cite{alexa2003differential,lipman2004differential,sorkine2004laplacian}. 

In our preliminary study \cite{nagata2018escherization}, we introduced a distance function, which focused on the similarity of the relative positional relationship of adjacent points between the two polygons, into the Escherization problem. This distance function is defined as follows: 
\begin{equation}
\label{eq:adjacency}
d_{AD}^2(U,W) = \sum_{i=1}^n { \left\| (\bm{\hat{u}_{i+1}}-\bm{\hat{u}_i}) - (\bm{\hat{w}_{i+1}}-\bm{\hat{w}_i})   \right\| }^2,  
\end{equation}
where $n+1$ is regarded as 1. We refer to $d_{AD}(U,W)$ as the {\it adjacent difference} (AD) distance. Fig. \ref{fig:normal_vs_adjacent} shows a typical example in which the AD distance is effective; the tile figure (b) is most similar to the goal polygon ``bat'' (figure (a)) in terms of the Euclidean distance and the tile figure (c) is most similar to the goal polygon in terms of the AD distance. Note that both tile shapes (red lines) are drawn on the goal polygon (black points) such that the Euclidean distance is minimized. 

Compared with the tile shape (b), the tile shape (c) does not overlap significantly with the goal polygon, but it appears to be intuitively more similar to the goal polygon. The reason is that the local shape of the contour of the goal polygon is well preserved in the tile shape (c) even though the overall structure is distorted  (e.g. the vertical width of the wings becomes narrower). As exemplified in this figure, even if the global structure is somewhat distorted, it would be better to preserve the local structures of the goal shape actively to search for more satisfactory tile shapes. 

\begin{figure} [t] %%%%%%%%%%%%%%%%%%%%%%%%%%%%%
\centering
\includegraphics[scale=0.30,keepaspectratio,clip]{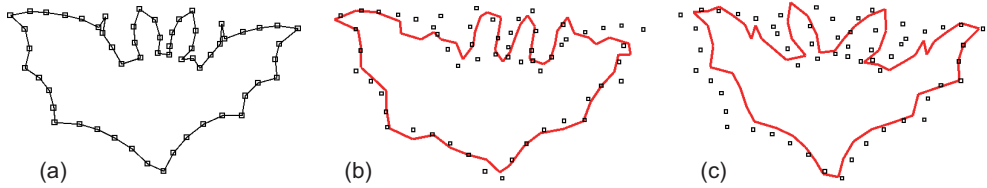}
\caption{(a) Goal polygon ``bat'' and tile shapes that are most similar to the goal polygon in terms of the Euclidean distance (b) and the AD distance (c).}
\label{fig:normal_vs_adjacent}
\end{figure} %%%%%%%%%%%%%%%%%%%%%%%%%%%%%

As a straightforward extension of the AD distance, the weighted AD (WAD) distance is defined by 
\begin{equation}
\label{eq:weight_AD}
d_{WAD}^2(U,W) = \sum_{i=1}^n k_i { \left\| (\bm{\hat{u}_{i+1}}-\bm{\hat{u}_i}) - (\bm{\hat{w}_{i+1}}-\bm{\hat{w}_i})   \right\| }^2,  
\end{equation}
where $n+1$ is regarded as 1. The values of $k_i \ (i = 1, 2, \dots, n)$ are weights assigned to the edges between the $i$th and $(i+1)$th points of the goal polygon $W$. The (W)AD distance can be defined as a special case of $d_G^2(U,W)$ and the matrix $G$ for the (W)AD distance is given by a $2n \times 2n$ symmetric matrix defined by 
$
G = 
\begin{pmatrix}
K & O \\
O & K \\
\end{pmatrix}
;
$
$K = (k(i,j))$ is an $n \times n$ matrix whose non-zero elements are given by 
\begin{equation}
\label{eq:G_WAD}
k(i,j) = 
\left\{
\begin{array}{lll}
-k_{i-1}        & (j=i-1) \\
k_{i-1} + k_i   & (j=i)   \  \ \ \ \ \ \ \ (i = 1, 2, \dots, n) \\ 
-k_i            & (j=i+1) \\
\end{array}%
\right.
,
\end{equation} 
where $j=n+1 \ (\mbox{resp.} \  -1)$ is regarded as $1 \ (\mbox{resp.} \ n)$.
Note that, when weights are introduced, the matrix $G$ depends on $j \in J$ (as in the case of the WE distance) because the indices of the weighted edges must be shifted depending on the start point of the numbering of the goal polygon $W$. Therefore, we need to define $G_j$ individually for the goal polygons $W_j \ (j=1, 2, \dots, n)$ when using the WAD distance.

\subsection{Generalized AD distance}  \label{sec:3_GAD}

We further design other distance functions that emphasize the similarity of the local structures between the two polygons $U$ and $W$ in different manners than the (W)AD distance. Unlike the idea of the (W)AD distance, it is also natural to emphasize the similarity of the relative positional relationship of {\it close} points (without restricting adjacent points), and two distance functions are designed from this point of view. 

Fig. \ref{fig:example_GAD} illustrates the concept of the proposed distance functions. For a goal polygon ``pegasus'' (figure (a)), we focus on the relative positional relationship of close points specified by the undirected edges (including those that constitute the goal polygon) illustrated in the figure (b) or (c). Let $E_c$ be a set of the selected edges, and the proposed distance function is defined as follows:
\begin{equation}
\label{eq:GAD}
d_{GAD}^2(U,W) = \sum_{(i,j) \in E_c} { \left\| (\bm{\hat{u}_i}-\bm{\hat{u}_j}) - (\bm{\hat{w}_i}-\bm{\hat{w}_j})   \right\| }^2. 
\end{equation}
We refer to $d_{GAD}(U,W)$ as the generalized AD (GAD) distance. 

We design two types of the GAD distance, which are referred to as GAD1 and GAD2. Let $d_{ave}$ be the average straight-line distance over the edges of the goal polygon $W$. For the GAD1 and GAD2 distances, $E_c$ is defined respectively as follows.
\begin{itemize}
\setlength{\itemindent}{15pt}   %5. 最初のインデント
\item[GAD1:]
The set $E_c$ is defined as the union of the edges of the goal polygon $W$ and the line segments between two points of $W$ whose lengths are less than $\gamma d_{ave}$ ($\gamma$ is a parameter), where only line segments inside the goal polygon $W$ are selected. If some of the selected line segments intersect, all intersections are removed by eliminating longer ones, except that the edges of the goal polygon are left with the highest priority. 
\item[GAD2:]
The set $E_c$ is defined in the same way as in the GAD1 distance, but line segments both inside and outside the goal polygon $W$ are selected. 
\end{itemize}

Fig. \ref{fig:example_GAD} shows examples of the sets $E_c$ for the GAD1 and GAD2 distances for a given goal polygon, where $\gamma = 1.4$. The GAD2 distance is motivated to consider the local structural similarity both inside and outside the goal polygon. We simply refer to the GAD1 and GAD2 distances as the GAD distance when they are not distinguished.

\begin{figure} [t] %%%%%%%%%%%%%%%%%%%%%%%%%%%%%
\centering
\includegraphics[scale=0.30,keepaspectratio,clip]{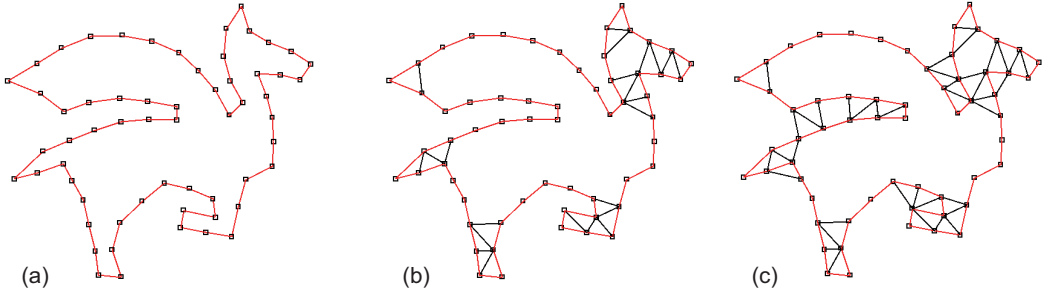}
\caption{(a) Goal polygon pegasus, (b) set $E_c$ of the GAD1 distance, and (c) set $E_c$ of the GAD2 distance ($\gamma=1.4$)}
\label{fig:example_GAD}
\end{figure} %%%%%%%%%%%%%%%%%%%%%%%%%%%%% 

The GAD distance can be represented as a special case of $d_G^2(U,W)$. Let $A$ and $D$ be the adjacency matrix and the degree matrix, respectively, of the graph defined by the points of the goal polygon and $E_c$. The matrix $G$ is then expressed as a $2n \times 2n$ symmetric matrix defined by 
$
G = 
\begin{pmatrix}
K & O \\
O & K \\
\end{pmatrix}
$
, where $K = D - A$ (i.e., the Laplacian matrix of the graph).
As in the cases of the WE and WAD distances, the matrix $G$ depends on $j \in J$, and we need to define $G_j$ individually for the goal polygons $W_j \ (j=1, 2, \dots, n)$.

\subsection{Exhaustive search of the templates} \label{sec:3_formulation}

From Eqs.~\ref{eq:u=Bxi} and \ref{eq:mahalanobis}, to perform the exhaustive search of the templates combined with the distance function $d_G^2(U,W)$, we need to solve the following optimization problem: 
\begin{equation}
\label{eq:formulation_general}
\argmin_{\bm{\xi}} \ {\bm \xi}^{\top} {B_{ik}}^{\top} G_j B_{ik} {\bm \xi} - 2{\bm w_j}^{\top} G_j B_{ik} {\bm \xi} + {\bm w_j}^{\top} G_j {\bm w_j},
\end{equation}
for all combinations of $i \in I$, $k \in K_i$, and $j \in J$. Here, we need to recall that (i) the matrix $B$ depends on $(i,k)$ and is denoted as $B_{ik}$, (ii) ${\bm w_j}$ represents the coordinates of the goal polygon $W_j$, and (iii) the matrix $G$ may depend on $j$ and is denoted as $G_j$. In Appendix A, we explain how this formulation is modified when the Procrustes distance must be used. 

The optimization problem (\ref{eq:formulation_general}) is simply the minimization of a quadratic function for each combination of $i \in I$, $k \in K_i$, and $j \in J$. We can obtain the solution to this optimization problem by solving the linear equation 
\begin{equation}
\label{eq:equation_general}
{B_{ik}}^{\top} G_j B_{ik} {\bm \xi} - {B_{ik}}^{\top} G_j {\bm w_j} = 0.
\end{equation}
We can solve this equation even if the matrix ${B_{ik}}^{\top} G_j B_{ik}$ is not regular, but it is easier and faster to solve it in the regular case. In fact, this matrix is not regular when the (W)AD and GAD distances are used. This is intuitively apparent because these distances are invariant to the translation of the tile shape $U$, and therefore the solution to Eq.~\ref{eq:equation_general} is not uniquely determined. A simple method to determine the solution uniquely is to add 1.0 (or any positive number) to $k(1, 1)$ (the top-left element of the matrix $K$), which is equivalent to add the square of the straight-line distance between the first points of $U$ and $W_j$ to the original distance function. Consequently, the solution is determined uniquely without changing the optimal value and the optimal tile shape, and therefore ${B_{ik}}^{\top} G_j B_{ik}$ becomes a regular matrix (symmetric positive-definite matrix in this case). We solved Eq.~\ref{eq:equation_general} using the Cholesky decomposition of the matrix ${B_{ik}}^{\top} G_j B_{ik}$ (it takes $O(n^3)$ time). Let ${\bm {\xi_{ikj}^*}}$ be the solution of Eq.~\ref{eq:equation_general}. The optimal value of the optimization problem (\ref{eq:formulation_general}) is then given by 
\begin{equation}
\label{eq:eval_ikj_general}
eval_{ikj}^G = -{\bm{\xi^*_{ikj}}}^{\top} {B_{ik}}^{\top} G_j B_{ik} \bm{\xi^*_{ikj}} + {\bm{w_j}}^{\top} G_j \bm{w_j}. 
\end{equation}

For each combination of $i \in I$ and $k \in K_i$, it takes $O(n^4)$ time for computing $eval_{ikj}^G$ for all $j \in J$ because it takes $O(n^3)$ time for solving Eq.~\ref{eq:equation_general}. If the WE, WAD, or GAD distance is used as a special case of $d_G^2(U,W)$, the time complexity is the same as in the general case. In contrast, if the AD distance is used, the matrix $G$ does not depend on $j \in J$ and we can reduce the time complexity. In this case, a set of the column vectors $\bm{b_1}, \bm{b_2}, \dots, \bm{b_m}$ consisting of the matrix $B_{ik}$ (see Eq.~\ref{eq:u=Bxi}) are linearly transformed into $\bm{b'_1}, \bm{b'_2}, \dots, \bm{b'_{m}}$ such that ${\bm b'_r}^{\top} G {\bm b'_s} = \delta_{rs}$ (the Kronecker delta function) for $r,s \in \{1, 2, \dots, m\}$. Such a set of column vectors can be obtained in $O(n^3)$ time by using the Gram-Schmidt orthogonalization process with an inner product defined as $<{\bm x}, {\bm y}> = {\bm x}^{\top} G {\bm y}$. Let a matrix $B'_{ik}$ be defined as $B'_{ik}=(\bm{b'_1}, \bm{b'_2}, \dots, \bm{b'_{m}})$ and the tile shape $U$ be parameterized by $\bm{u} = B'_{ik} \bm{\xi}$. The solution of Eq.~\ref{eq:equation_general} with $B_{ik}$ replaced with $B'_{ik}$ is then computed by ${\bm {\xi_{ikj}^*}} =  {B'_{ik}}^{\top} G {\bm w_j}$ for $j \in J$ (because ${B'_{ik}}^{\top} G_j B'_{ik}$ is the identity matrix). Therefore, for the AD distance, it takes $O(n^3)$ time for computing $eval_{ikj}^G$ for all $j \in J$.

\section{Efficient exhaustive search algorithms} \label{sec:4}

When the WE, AD, WAD, and GAD distances are introduced, the exhaustive search of the templates, i.e., solving the optimization problem (\ref{eq:formulation_general}) for all combinations of $i \in I$, $k \in K_i$, and $j \in J$, is very time-consuming. In this section, we propose efficient exhaustive search algorithms for these distance functions. Recently, Nagata and Imahori \cite{nagata2019} developed an efficient exhaustive search algorithm for the Euclidean (and the Procrustes) distance case. The efficient exhaustive search algorithms developed for the proposed distance functions are partially based on this algorithm. Therefore, the outline of this algorithm is first described and then the efficient exhaustive search algorithms developed for other distance functions are presented. 

As in Section \ref{sec:3}, the necessary explanation is provided based on the case where the Euclidean distance is available. However, the structure of the algorithms presented in this section does not change when the Procrustes distance must be used. 

\subsection{The Euclidean distance case} \label{sec:4_Euclidean}

When the Euclidean distance is used, the exhaustive search of the templates, i.e., solving the optimization problem (\ref{eq:formulation_ikj_Euclid}) for all combinations of $i \in I$, $k \in K_i$, and $j \in J$, can be performed very efficiently using the following three fundamental techniques \cite{nagata2019}. 
\begin{enumerate}
\item
An $O(n)$ time construction method of the matrix $B_{ik}$. 
\item
An $O(n)$ time calculation method of $eval^E_{ikj}$. 
\item
An efficient exhaustive search algorithm, which skips unnecessary computation of $eval^E_{ikj}$. 
\end{enumerate}
Owing to these techniques, the developed algorithm took only 0.55 s and 9.01 s for 60 and 120-point goal polygons, respectively, to perform the exhaustive search on a standard modern PC. 

Regarding the first technique, if the orthogonality between the column vectors of $B_{ik}$ is not required, it is not difficult to construct $B_{ik}$ in $O(n)$ time \cite{nagata2019}. This orthogonality is not required for the exhaustive search algorithms developed for the proposed distance functions.

The second technique was developed for the Procrustes distance. When the Euclidean distance is applicable to the selected isohedral type, it is almost trivial to compute $eval^E_{ikj}$ (Eq.~\ref{eq:eval_ikj_Euclid}) in $O(n)$ time using an appropriate sparse matrix format because the matrix $B_{ik}$ constructed using the first technique is a sparse matrix whose order of non-zero elements is $O(n)$. 

Regarding the third technique, an efficient method of calculating a lower bound on the value of $eval^E_{ikj}$ was developed to skip unnecessary computation of $eval^E_{ikj}$. This technique was applied to three isohedral types IH4, IH5, and IH6 because the order of $K_i$ is greater than or equal to $O(n^3)$ only for these isohedral types, and the exhaustive search of the templates spends most of the computation time for these isohedral types. Algorithm \ref{alg:E} depicts an outline of the efficient exhaustive search algorithm. 

\begin{algorithm}[t]
\caption{\sc Outline of the efficient exhaustive search with the Euclidean distance}
\label{alg:E}
\nl $eval_{mim} := \infty$\;
\nl \For{$i \in I$ (the nine most general isohedral types)}{
  \nl \For{$k \in K_i$}{
    \nl \For{$j = 1, 2, \dots, n$}{ 
      \nl $eval^E := \infty$\;
      \nl \If{$i = 4,5, \ \mbox{or} \ 6$}{
        \nl Compute a lower bound on $eval^E_{ikj} \rightarrow L$\;
        \nl \If{$L < eval_{min}$}{
          \nl Compute $eval^E_{ikj} \rightarrow eval^E$\;
        }
      }
      \nl \Else{
        \nl Compute $eval^E_{ikj} \rightarrow eval^E$\;
      }
      \nl \If{$eval^E < eval_{min}$}{ 
        \nl $eval_{min} := eval^E$\;
        \nl Update the current best tile shape $\bm{u^*}$\;
      }
    }
  }
}
\nl \Return $\bm{u^*}$\;
\end{algorithm}

In Algorithm \ref{alg:E}, the value of $eval^E_{ikj}$ is computed only when necessary (line 9) for IH4, IH5, and IH6 because there is no possibility of updating the current best tile shape $\bm{u^*}$ if the computed lower bound $L$ (line 7) is greater than or equal to the current best value $eval_{mim}$. For a complete understanding of this algorithm, we need to know how to compute a lower bound on $eval^E_{ikj}$ efficiently. For further details, see Section 4 of \cite{nagata2019}. 

% A more detailed explanation of Algorithm 1 is presented in Appendix B. 

\subsection{The WE distance case} \label{sec:4_weight}

When the WE distance (Eq.~\ref{eq:euclid_weight}) is used, we can construct an efficient algorithm for the exhaustive search of the templates by utilizing Algorithm 1 in a simple way. Recall that, when the WE distance is used, the matrix $G \ (=G_1)$ is the $2n \times 2n$ diagonal matrix whose diagonal elements are given by $k_1, \dots, k_n, k_1, \dots, k_n$, and $G_j$ is obtained from $G$ by shifting the diagonal elements accordingly (see Section \ref{sec:3_WE}).

Let $eval^W_{ikj}$ be the optimal value of the optimization problem (\ref{eq:formulation_general}) for each combination of $i \in I$, $k \in K_i$, and $j \in J$ when the WE distance is used. The basic idea is to utilize $eval^E_{ikj}$ as a lower bound on $eval^W_{ikj}$, and we have the following theorem.
\begin{theorem}
\label{th:1}
$eval^E_{ikj} \leq eval^W_{ikj}$ for every combination of $i \in I$, $k \in K_i$, and $j \in J$, if all weights $k_1, k_2, \dots, k_n$ are greater than or equal to one. 
\end{theorem}
\begin{proof}
The value of $eval^E_{ikj}$ is the minimum value of $({\bm u} - {\bm w_j})^{\top} E ({\bm u} - {\bm w_j})$ under the constraint ${\bm u} = B_{ik} {\bm \xi}$, where $E$ is the identity matrix. The value of $eval^W_{ikj}$ is the minimum value of $({\bm u} - {\bm w_j})^{\top} G_j ({\bm u} - {\bm w_j})$ under the same constraint. Let $G_j$ be decomposed into $E+A$, where $A$ is a diagonal matrix whose diagonal elements are zero or above. Given that both $E$ and $A$ are positive semi-definite matrices, this theorem is trivial. 
\end{proof}
\noindent
%The value of $eval^E_{ikj}$ is expected to provide a better lower bound on $eval^W_{ikj}$ as the total amount of additional weights decreases.

Using Theorem \ref{th:1}, an efficient exhaustive search algorithm combined with the WE distance is constructed based on Algorithm \ref{alg:E}. Algorithm \ref{alg:WE} depicts an outline of this algorithm; for each triplet $(i,k,j)$, the value of $eval_{ikj}^{W}$ is computed only when $eval_{ikj}^E$ (and its lower bound $L$) is less than the current best value $eval_{mim}$ (line 15). By using Algorithm \ref{alg:WE}, the number of computations of $eval^W_{ikj}$ is reduced drastically, especially when the total amount of weights that exceed 1.0 is small.

\begin{algorithm}[t]
\caption{\sc Outline of the efficient exhaustive search with the WE distance}
\label{alg:WE}
\nl $eval_{mim} := \infty$\;
\nl \For{$i \in I$}{
  \nl \For{$k \in K_i$}{
    \nl \For{$j = 1, 2, \dots, n$}{ 
      \nl $eval^E := \infty$, $eval^W := \infty$\;
      \nl \If{$i = 4,5, \ \mbox{or} \ 6$}{
        \nl Compute a lower bound on $eval^E_{ikj} \rightarrow L$\;
        \nl \If{$L < eval_{min}$}{
          \nl Compute $eval^E_{ikj} \rightarrow eval^E$\;
        }
      }
      \nl \Else{
        \nl Compute $eval^E_{ikj} \rightarrow eval^E$\;
      }
      \nl \If{$eval^E < eval_{min}$}{
        \nl Compute $eval^W_{ikj} \rightarrow eval^W$;
      }
      \nl \If{$eval^W < eval_{min}$}{ 
        \nl $eval_{min} := eval^W$\;
        \nl Update the current best tile shape $\bm{u^*}$\;
      }
    }
  }
}
\nl \Return $\bm{u^*}$\;
\end{algorithm}

\subsection{The AD distance case} \label{sec:4_AD}

When the AD distance (Eq.~\ref{eq:adjacency}) is used, we can construct an efficient algorithm for the exhaustive search of the templates based on Algorithm 1, but in a different way from that described in the previous subsection. 

Let $eval^{AD}_{ikj}$ be the optimal value of the optimization problem (\ref{eq:formulation_general}) for each combination of $i \in I$, $k \in K_i$, and $j \in J$ when the AD distance is used. Recall that the $xy$-coordinates of the tile polygon $U$ and the goal polygon $W$ are represented as $2n$-dimensional vectors $\bm{u} = {(x_1, x_2, \dots, x_n, y_1, y_2, \dots, y_n)}^{\top}$ and $\bm{w} = {(x^w_1, x^w_2, \dots, x^w_n, y^w_1, y^w_2, \dots, y^w_n)}^{\top}$. Then, we define $2n$-dimensional vectors $\overline{\bm u} = {(x_2-x_1, x_3-x_2, \dots, x_1-x_{n},}$ ${ y_2-y_1,y_3-y_2, \dots, y_1-y_n)}^{\top}$ and $\overline{\bm w} = {(x^w_2-x^w_1, x^w_3-x^w_2,}$ ${\dots, x^w_1-x^w_{n}, y^w_2-y^w_1, y^w_3-y^w_2, \dots, y^w_1-y^w_n)}^{\top}$. Let $\bm{\hat{\overline{u}}}_i$
and $\bm{\hat{\overline{w}}}_i$ be defined as $(x_{i+1}-x_i, y_{i+1}-y_i)^{\top}$ and $(x^w_{i+1}-x^w_i, y^w_{i+1}-y^w_i)^{\top}$, respectively, for $i = 1, 2, \dots n$, where $n+1$ is regarded as 1. The AD distance is then represented by
\begin{equation}
\label{eq:AD_transformed}
d_{AD}^2(U,W) = \sum_{i=1}^n { \left\| \bm{\hat{\overline{u}}_i} -  \bm{\hat{\overline{w}}_i} \right\| }^2 = {\left\| \overline{\bm u} - \overline{\bm w} \right\| }^2. 
\end{equation}
The right-hand side of this expression is equivalent to the definition of the Euclidean distance (Eq.~\ref{eq:euclid}) if $\overline{\bm u}$ and $\overline{\bm w}$ are replaced with ${\bm u}$ and ${\bm w}$, respectively. Therefore, if $\overline{\bm u}$ can be parameterized as $\overline{\bm u} = {\overline{B}_{ik}} {\bm \xi}$ with an appropriate matrix ${\overline{B}_{ik}}$ whose column vectors are mutually orthonormal as in Eq.~\ref{eq:u=Bxi}, the optimization problem (\ref{eq:formulation_general}) can be converted into the form of the optimization problem (\ref{eq:formulation_Euclid}) and we can compute $eval^{AD}_{ikj}$ as in Eq.~\ref{eq:eval_ikj_Euclid}. That is, we have $eval^{AD}_{ikj} = -{\bm{\xi^*_{ikj}}}^{\top} \bm{\xi^*_{ikj}} + {\overline{\bm w}}^{\top} \overline{\bm w}$, where $\bm{\xi^*_{ikj}} = {\overline{B}_{ik}}^{\top} \bm{\overline{w}_j}$. Here, $\bm{\overline{w}_j}$ is defined in a similar way as $\bm{w_j}$, i.e., it is obtained from $\overline{\bm w}$ by shifting the elements accordingly.

If the vector $\overline{\bm u}$ can be parameterized in the manner described above, we can directly apply Algorithm 1 to perform the exhaustive search with the AD distance, thus taking advantage of the efficiency of this algorithm. Owing to the third technique of Algorithm \ref{alg:E} (see Section \ref{sec:4_Euclidean}), we can skip unnecessary computation of $eval^{AD}_{ikj}$ because a lower bound on $eval^{AD}_{ikj}$ can be efficiently computed. 

The efficient exhaustive search algorithm with the AD distance is summarized as follows. 

\smallskip
\noindent
{\bf [Algorithm 3] Efficient exhaustive search with the AD distance} 
\begin{itemize}
\item
Algorithm 1 is performed by replacing the definitions of $\bm u$, $\bm w$, $\bm{w_j}$, and $B_{ik}$ with those of $\bm{\overline{u}}$, $\bm{\overline{w}}$, $\bm{\overline{w}_j}$, and $\overline{B}_{ik}$, respectively.  
\end{itemize}

%Note that the best tile shape $\bm{\overline{u}}^*$ obtained using Algorithm 3 represents the coordinates of the tile polygon $U$ indirectly. We need to transform $\bm{\overline{u}}^*$ into the coordinates of the corresponding tile polygon by setting $x_1^*$ and $y_1^*$ to any value. 

%The matrix $\overline{B}_{ik}$ can be obtained as follows. Recall that, in the case of the Euclidean distance, ${\bm u}$ is parameterized as ${\bm u}= B_{ik} {\bm \xi}$ for every combination of $i \in I$ and $k \in K_i$. Here, the matrix $B_{ik}$ is a $2n \times m$ sparse matrix whose column elements (vectors) are orthonormal to each other and whose order of non-zero elements is $O(n)$ (see Section \ref{sec:4_Euclidean}). 

For each combination of $i \in I$ and $k \in K_i$, the possible values for the vector $\overline{\bm u}$ can be parameterized as follows (indices $i$ and $k$ are omitted for simplicity).  Let $B' \ (=B'_{ik})$ be the $2n \times m$ matrix obtained from $B \ (=B_{ik})$ by shifting down the first $n$th rows and the last $n$th rows, respectively, one by one cyclically. Then, $\overline{\bm u}$ is parameterized by $\overline{\bm u} = (B'-B) {\bm \xi}$. However, the column elements of the matrix $(B'-B)$ are not mutually orthogonal and we need to modify this matrix.

Let $\bm{b'_1}, \bm{b'_2}, \dots, \bm{b'_m}$ be a set of the column vectors of the matrix $({B'}-B)$. Then, these vectors are linearly transformed into $\bm{\overline{b}_1}, \bm{\overline{b}_2}, \dots, \bm{\overline{b}_{m}}$ such that they are mutually orthonormal. Then, the matrix $\overline{B} \ (= {\overline{B}_{ik}})$ is obtained as $\overline{B}=(\bm{\overline{b}_1} \ \bm{\overline{b}_2} \ \dots \ \bm{\overline{b}_{m}})$. The Gram--Schmidt orthonormalization procedure was used to obtain the vectors $\bm{\overline{b}_1}, \bm{\overline{b}_2}, \dots, \bm{\overline{b}_{m}}$. Therefore, it takes $O(n^3)$ time to compute $\overline{B}_{ik}$ (because $m \sim O(n)$).

\subsection{The WAD and GAD distance cases} \label{sec:4_GAD}

When the WAD distance (Eq.~\ref{eq:weight_AD}) or the GAD distance (Eq.~\ref{eq:GAD}) is used, we can construct an efficient algorithm for the exhaustive search of the templates based on Algorithms 2 and 3. Both cases are explained together. 

Let $eval^{WAD}_{ikj}$ ($eval^{GAD}_{ikj}$) be the optimal value of the optimization problem (\ref{eq:formulation_general}) for each combination of $i \in I$, $k \in K_i$, and $j \in J$ when the WAD (GAD) distance is used. Then, we have the following theorems: 
\begin{theorem}
\label{th:2}
$eval^{AD}_{ikj} \leq eval^{WAD}_{ikj}$ for every combination of $i \in I$, $k \in K_i$, and $j \in J$, if all weights $k_1, k_2, \dots, k_n$ are greater than or equal to one. 
\end{theorem}
\begin{proof}
Let $G_j$ defined for the WAD distance be decomposed into $E+A$, where $E$ is the matrix $G$ defined for the AD distance. From the definitions of the AD and WAD distances, both $E$ and $A$ are positive semi-definite matrices. Therefore, this theorem is proved in the same way as in Theorem \ref{th:1}. 
\end{proof}

\begin{theorem}
\label{th:3}
$eval^{AD}_{ikj} \leq eval^{GAD}_{ikj}$ for every combination of $i \in I$, $k \in K_i$, and $j \in J$.
\end{theorem}
\begin{proof}
Same as the proof of Theorem \ref{th:2}. 
\end{proof}

Considering that the value of $eval^{AD}_{ikj}$ can be used as a lower bound on the value of $eval^{WAD}_{ikj}$ ($eval^{GAD}_{ikj}$), an efficient exhaustive search algorithm with the WAD (GAD) distance is constructed in a similar way as Algorithm 2. Algorithm 4 depicts an efficient exhaustive search algorithm for the WAD distance. The same algorithm is applicable to the GAD distance simply by replacing WAD with GAD in the algorithm. 

\smallskip
\noindent
{\bf [Algorithm 4] Efficient exhaustive search with the WAD distance} 
\begin{itemize}
\item
Algorithm 2 is performed by replacing some items as follows: 
\begin{enumerate}
\item[(a)]
$eval^E_{ikj}, eval^E \rightarrow eval^{AD}_{ikj}, eval^{AD}$.
\item[(b)]
$eval^W_{ikj}, eval^W  \rightarrow eval^{WAD}_{ikj}, eval^{WAD}$.
\end{enumerate}
\item
The value of $eval^{AD}_{ikj}$ and its lower bound are computed according to Algorithm 3. 
\end{itemize}

When the WAD distance is used, the number of computations of $eval^{WAD}_{ikj}$ is reduced drastically, especially when the total amount of weights that exceed 1.0 is small. This is also true for the GAD distance, especially when the value of the parameter $\gamma$ is small. However, when the GAD distance was used in our preliminary experiments, executing Algorithm 4 required considerable computation times at appropriate values of $\gamma$ (1.2--1.6) for goal polygons composed of relatively many points ($n > 90$). The reason is that the value of $eval_{ikj}^{AD}$ no longer provides a tight lower bound on $eval^{GAD}_{ikj}$ because the set $E_c$ includes a significant number of edges other than those of the goal polygon $W$. 

To obtain results within an acceptable computation time when using the GAD distance, we present an incomplete search algorithm, which searches only promising configurations of $(k,j) \in (K_i, J)$ for each isohedral type $i \in I$. The basic idea is to estimate a lower bound on $eval^{GAD}_{ikj}$ heuristically. From the definition of the GAD distance (Eq.~\ref{eq:GAD}), it can be expected that the value of $eval^{GAD}_{ikj}$ is roughly proportional to the sum of the lengths of the edges in $E_c$. Let $d_{ij}$ be the straight-line distance between two points $i$ and $j$ on the goal polygon $W$ and $E_0$ be a set of the edges of $W$. A heuristic lower bound on $eval^{GAD}_{ikj}$ is then estimated by 
$r \cdot eval^{AD}_{ikj}$, where $r = \alpha \frac{\sum_{(i,j) \in E_c} d_{ij}}{\sum_{(i,j) \in E_0} d_{ij}}$ and $\alpha$ is a parameter. With the heuristic lower bound on $eval^{GAD}_{ikj}$, an incomplete search algorithm (Algorithm 5) for the GAD distance is constructed in the same way as Algorithm 4.

\smallskip
\noindent
{\bf [Algorithm 5] Incomplete search with the GAD distance} 
\begin{itemize}
\item
Algorithm 4 is performed by replacing the condition of the if-statement in line 14 (i.e., $eval^{AD} < eval_{min}$) with $r \cdot eval^{AD} < eval_{min}$. 
\end{itemize}

The heuristic lower bound may be greater than $eval^{GAD}_{ikj}$, and therefore the optimal tile shape with the minimum value of $eval^{GAD}_{ikj}$ (over all possible combinations of $i \in I$, $k \in K_i$, and $j \in J$) may be overlooked. The parameter $\alpha$ was set to 0.9 in consideration of the trade-off between the required computation time and the ratio of overlooking the optimal solution (in fact, the top ten solutions as described later).

\section{Experimental Results} \label{sec:5}

Experiments were conducted to examine the impact of the proposed distance measures on the quality of tile shapes when these distance measures were incorporated into the exhaustive search of the templates. The effectiveness of the efficient exhaustive (and incomplete) search algorithms combined with these distance measures was also verified. 

\subsection{Experimental settings} \label{sec:5_setting}

All the algorithms were implemented in C++ with the Eigen library (template headers for linear algebra, matrix, and vector operations) in Ubuntu 14.04 Linux. The programs were executed on PCs, each with an Intel Core i7-4790 3.60~GHz CPU.

The six distance measures considered in this study are summarized in Table \ref{table:distance_list} along with the corresponding efficient exhaustive (incomplete) search algorithms. For each algorithm, the top ten tile shapes in terms of the corresponding distance measure were stored by modifying the algorithm; this is easily implemented simply by defining $eval_{min}$ as the distance value of the current tenth best solution. We then selected an intuitively best one among them (see Section \ref{sec:5_quality}). In addition to Algorithms 2--4, for each distance measure, a naive exhaustive search algorithm was performed to assess the efficiency of Algorithms 2--4, where the value of ${eval_{ikj}^G}$ (Eq.~\ref{eq:eval_ikj_general}) was computed for all combinations of $i \in I$, $k \in K_i$, and $j \in J$. As for the WE (WAD) distance, weight 4 was assigned to 10--15 points (9--14 edges) of goal polygons and these points (edges) are specified later.

The efficient exhaustive (incomplete) search algorithms and the naive algorithms were applied to the six goal figures, bat ($n=58$), seahorse ($n=60$), pegasus ($n=60$), squid ($n=92$), octopus ($n=94$), and spider ($n=120$), shown in Fig. \ref{fig:result1}. 

\begin{table}
\begin{center}
\caption{Summary of the six distance measures and the corresponding efficient exhaustive (incomplete) search algorithms}
\label{table:distance_list}
\begin{tabular}{c|c|c}
\hline	
Distance   &   Parameters              &   Exhaustive (Incomplete)  \\ 
measures   &                           &   search \\ 
\hline	
\hline	
Euclidean  &    --                    & Algorithm 1 \\
\hline	
 WE        &  weights of vertices     & Algorithm 2 \\
\hline	
 AD        &    --                    & Algorithm 3 \\
\hline	
 WAD       &  weights of edges        & Algorithm 4 \\
\hline	
     GAD1  & $\gamma =$ 1.2, 1.4, or 1.6 & Algorithm 4 (5) \\
\hline	
     GAD2  & $\gamma =$ 1.2, 1.4, or 1.6 & Algorithm 4 (5) \\ 
\hline	
\end{tabular}
\end{center}
\end{table}

\subsection{Quality of the tile shapes} \label{sec:5_quality}

Fig. \ref{fig:result1} shows the tile shapes for the six goal figures obtained with the six distance measures using the corresponding exhaustive search algorithms. For each goal figure, Fig. \ref{fig:result1} shows the goal polygon, the set $E_c$ for the GAD2 distance (only the line segments inside the goal polygons are considered for the GAD1 distance), and intuitively best tile shapes obtained with the six distance measures. Note that, for each distance measure, the tile shape with the minimum distance value was not necessarily intuitively the most satisfactory one. Therefore, we selected an intuitively best one from the top ten tile shapes, where the numbers in parentheses indicate the ranking of the distance value. % All the top ten tile shapes obtained with each distance measure for the six goal figures are presented in the supplementary file. 

\begin{figure} [t] %%%%%%%%%%%%%%%%%%%%%%%%%%%%%
\centering
\includegraphics[scale=0.24,keepaspectratio,clip]{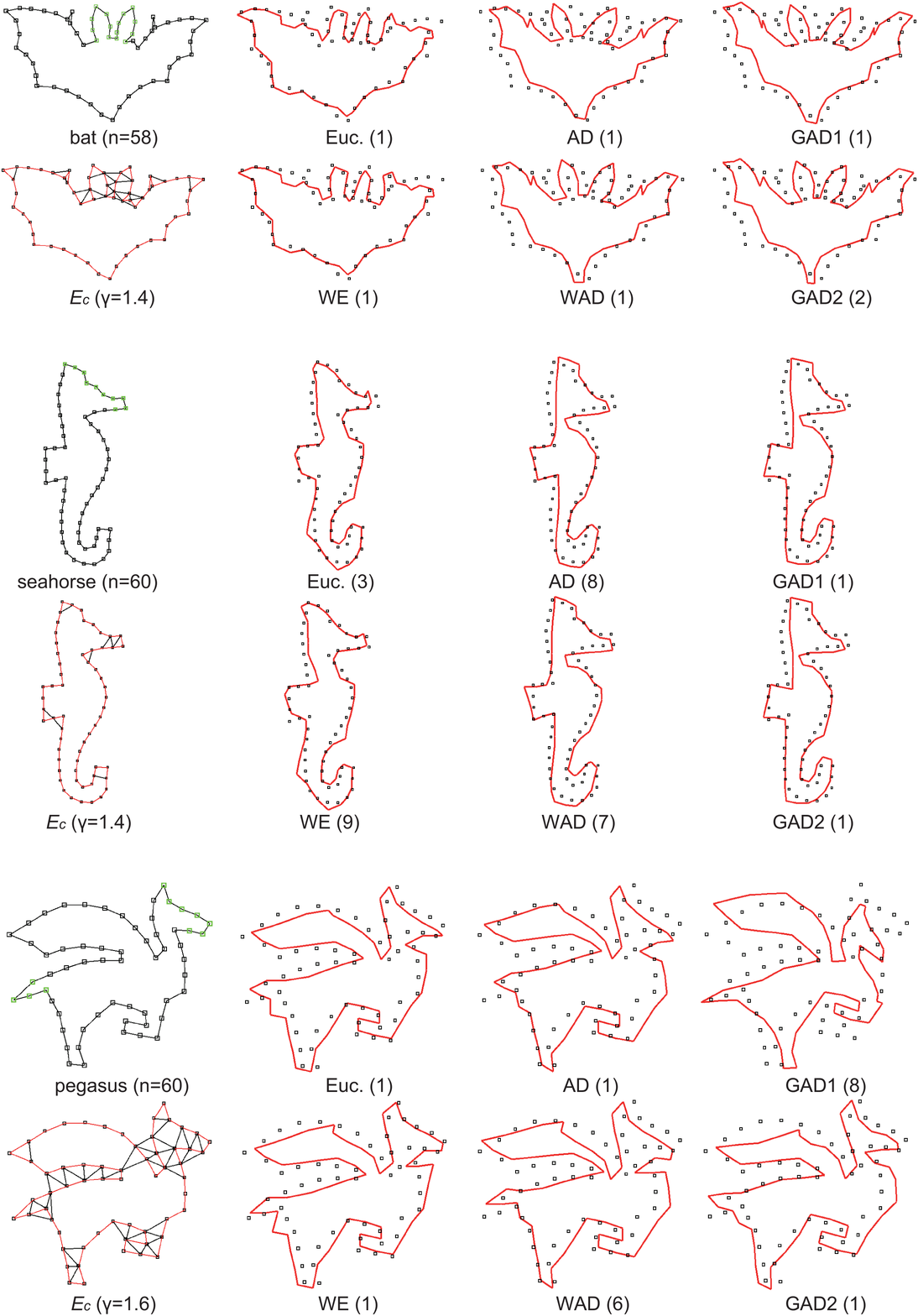}
\caption{Best tile shapes obtained with the six distance measures. For each distance measure, an intuitively best tile shape is selected from the top ten tile shapes, where the numbers in parentheses indicate the ranking of the distance value.}
\label{fig:result1}
\end{figure} %%%%%%%%%%%%%%%%%%%%%%%%%%%%%
\addtocounter{figure}{-1}
\begin{figure} [t] %%%%%%%%%%%%%%%%%%%%%%%%%%%%%
\centering
\includegraphics[scale=0.235,keepaspectratio,clip]{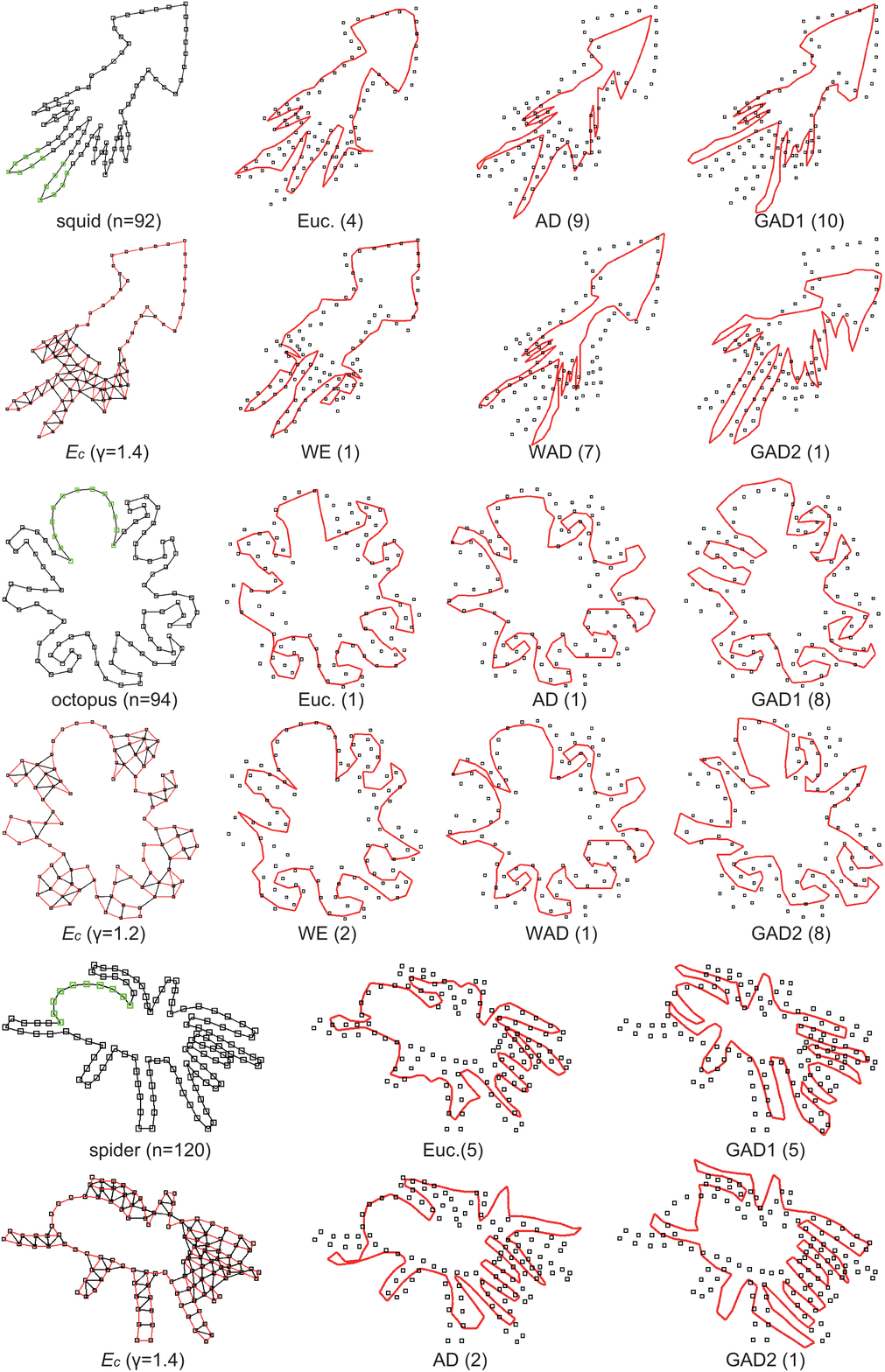}
\caption{(continued)}
\end{figure} %%%%%%%%%%%%%%%%%%%%%%%%%%%%%

\begin{figure*} [t] %%%%%%%%%%%%%%%%%%%%%%%%%%%%%
\centering
\includegraphics[scale=0.23,keepaspectratio,clip]{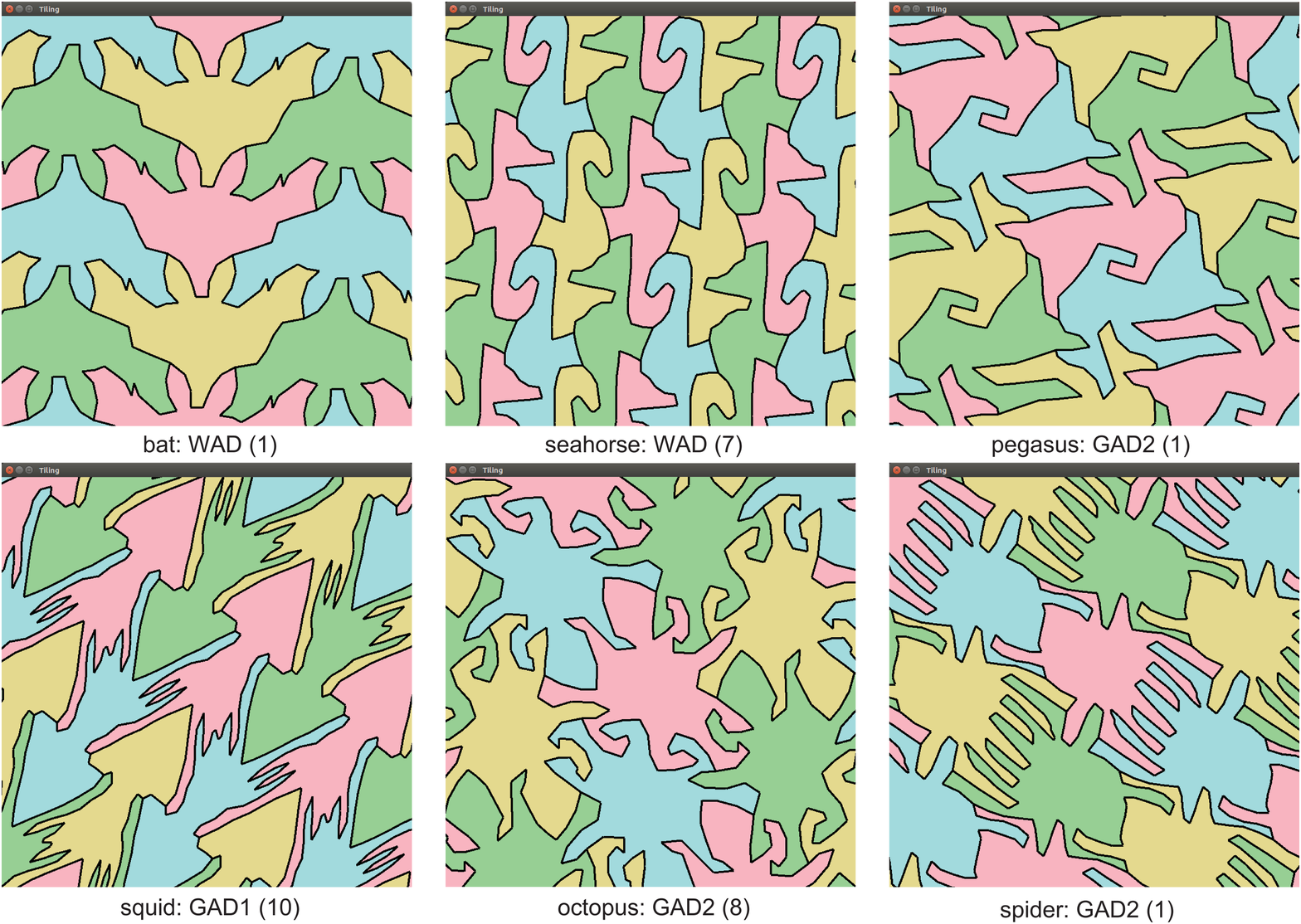}
\caption{Examples of tilings. Each tiling is formed from one of the tile shapes displayed in Fig. \ref{fig:result1}.}
\label{fig:tiling}
\end{figure*} %%%%%%%%%%%%%%%%%%%%%%%%%%%%%

In Fig. \ref{fig:result1}, each tile shape is superimposed onto the goal polygon (black points) such that the Euclidean distance between the two is minimized. The goal polygons shown in the figure represent information about the weights for the WE and WAD distances, where weight 4 is assigned to the green points (edges between the green points) for the WE (WAD) distance. As for the GAD distances, for each goal polygon, only the results obtained with the value of $\gamma$ specified in the figure are presented, where the value of $\gamma$ was selected such that relatively good results were obtained. 

We first focus on the tile shapes obtained with the Euclidean distance. We can observe that the obtained tile shapes overlap with the goal polygons most compared with the other cases. This is a natural consequence of the definition of the Euclidean distance. However, for complex goal polygons squid, octopus, and spider, the difference between the obtained tile shapes and the goal polygons are partly noticeable. 

Subsequently, we focus on the tile shapes obtained with the AD distance. The result for goal polygon bat is a typical example where the tile shape obtained with the AD distance is intuitively more satisfactory than the one obtained with the Euclidean distance. The reason is that the contour of the goal polygon is well preserved in the former one although the positions of the vertices of the tile polygon deviate significantly from the corresponding positions of the goal polygon. The result for goal polygon seahorse also appears to be more satisfactory than the one obtained with the Euclidean distance for the same reason. For the other goal polygons, the tile shapes obtained with the AD distance appear to be more or less satisfactory than the ones obtained with the Euclidean distance. 

Subsequently, we focus on the tile shapes obtained with the WE and WAD distances. Results for goal polygon spider are omitted here to enlarge the tile shapes obtained with the other distance measures. Basically, the tile shapes obtained with the WE (WAD) distance are similar to those obtained with the Euclidean (AD) distance, but the differences from the goal polygons in the weighted parts become smaller as expected. The results for goal polygons seahorse, pegasus, and octopus are typical examples in which the use of the WE and WAD distances improves the satisfaction of the obtained tile shapes. That is, the shapes of the heads, which would be most important parts characterizing these goal figures, are better preserved than those in the Euclidean and AD distance cases. For goal polygons squid and spider, it was difficult to determine appropriate parts to weight, and weights were assigned to the selected parts for no particular reason. 

Subsequently, we focus on the tile shapes obtained with the GAD1 distance. For relatively simple goal polygons bat and seahorse, the obtained tile shapes are not significantly different from the ones obtained with the (W)AD distance. Typical successful (or interesting) examples for the GAD1 distance are found in the results for goal polygons pegasus and squid, where intuitively satisfactory tile shapes are obtained by moderate deformations of the goal polygons while maintaining the local structures. For example, as expected from the result for squid, it is necessary to open the angle of the two tentacles to obtain a satisfactory tile shape, where the head is inserted between the tentacles in the resulting tiling (see Fig. \ref{fig:tiling}). As for the results for goal polygons octopus and spider, the local structures (especially the thickness of the legs and arms) are well maintained in the obtained tile shapes, as expected from the definitions of the GAD1 distance. For this reason, a relatively satisfactory tile shape is obtained for octopus. However, the tile shape obtained for spider\footnote{Almost the same tile shape as that displayed for the GAD2 distance was also obtained, but we selected this tile shape to display because it is interesting (and was not found with the GAD2 distance) and similar tile shapes were typically obtained with the GAD1 distance.} may no longer look like a spider because the shape of the abdomen was shrunk too much. 

%Note that the two tentacles are rotated about 10 degrees to form the tile shape, even though the GAD1 distance is not invariant to local rotations of the tile shape. This suggests that the local rigidity of the tile shape can be evaluated appropriately using the GAD1 distance if the degree of local rotations is not large. 

Subsequently, we focus on the tile shapes obtained with the GAD2 distance. Similar to the results of the GAD1 distance, the local structures are well maintained in the obtained tile shapes, but the overall structure is fairly different from when the GAD1 distance is used for relatively complex goal polygons. From the definition of the GAD2 distance, it is expected that tile shapes obtained with the GAD2 distance are less deformed compared to those obtained with the GAD1 distance because the GAD2 distance considers the local structural similarity both inside and outside the goal polygon. This has both positive and negative effects, which make it possible to create more satisfactory tile shapes for some goal polygons. For pegasus and spider, the tile shapes obtained with the GAD2 distance seems to be more satisfactory than the ones obtained with the GAD1 distance. For squid, however, unlike the case of the GAD1 distance, the tile shape obtained with the GAD2 distance is far from satisfactory. The reason is that it is necessary to open the angle of the two tentacles to obtain a satisfactory tile shape (as explained earlier), but the GAD2 distance is sensitive to such a deformation.

%From the observation of results for various goal polygons including those other than the presented six goal polygons, the use of the GAD1 distance tends to create more or less satisfactory tile shapes compared with the use of the GAD2 distance. This is a natural consequence because local structures of goal polygons are determined by the relative positional relationship of close points not only inside but also outside the figure. However, for the goal polygon squid, the tile shape obtained with the GAD2 distance is most satisfactory. As expected from the tile shape obtained with the GAD2 distance, it is necessary to open the angle of the two tentacles to obtain a satisfactory tile shape, where the head of the squid is inserted between the tentacles in the tiling. From the definition of GAD2 distance, the distance value is less variable for such a movement (compared with the GAD1 distance), and this tile shape was highly evaluated in terms of the GAD2 distance.

Finally, Fig. \ref{fig:tiling} shows the tilings for the six goal polygons, each of which is created from an intuitively best tile shape selected from the tile shapes shown in Fig. \ref{fig:result1}.

\subsection{Efficiency of the algorithms} \label{sec:5_efficiency}

Table \ref{table:efficiency} lists the execution times of the efficient exhaustive (incomplete) search algorithms combined with the Euclidean, WE, AD, WAD, GAD1, and GAD2 distances. The results of the naive exhaustive search algorithms are also listed in the table. 

For the Euclidean distance, Algorithm 1 performs the exhaustive search very efficiently. For other distance measures, the execution times of the naive algorithms are much greater than that of Algorithm 1. More specifically, the naive algorithm for the AD distance is approximately 18 to 120 times slower than Algorithm 1 on the six goal polygons. The execution times of the naive algorithms for the WE, WAD, and GAD distances are roughly the same on each of the six goal polygons and are approximately 400--3,000 times greater than that of Algorithm 1; note that the execution times for the GAD1 and GAD2 distances are omitted in the table (about 1.1 times greater than the execution time of the WAD distance). The table shows that, for the WE, AD, and WAD distances, the efficient exhaustive search algorithms (Algorithms 2--4) are considerably faster than the naive algorithms, which makes it possible to perform the exhaustive search in a realistic computation time even for 100-point goal polygons. For the GAD1 and GAD2 distances, the efficient exhaustive search algorithm (Algorithm 4) is faster than the naive algorithm, but the speed-up effect is not so remarkable as in the case of the WE and WAD distances, especially for relatively large goal polygons ($n \geq 92$). Therefore, the required execution time may not be tolerable for these goal polygons. Instead, for the GAD1 and GAD2 distances, the proposed incomplete search algorithm (Algorithm 5) is useful to obtain results in a reasonable computation time, although there is a risk that some of the top ten tile shapes are overlooked. However, we confirmed that the top ten tile shapes were rarely overlooked; only the 10th best tile shape for the GAD2 ($\gamma=1.6$) distance on goal polygon seahorse was overlooked in the experiments. 

% Note that the naive algorithms also used some speed-up techniques developed in Algorithm 1 (especially the first technique), and the execution time of the naive algorithms will increase significantly if these techniques are not used\footnote{In \cite{imahori2015escher}, it took 4.5 s for a 60-point goal polygon when the weighted Procrustes distance was used even though only one configuration of $K_i$ was considered for each isohedral type.}. 

%For the WE and WAD distances, the execution time depends on the magnitude of weights and the number of weighted points, but the weight settings used in the experiment would cover the settings in standard use. 

\begin{table}
\begin{center}
\caption{Execution time (in seconds, or hours if specified) of the efficient exhaustive (incomplete) search algorithms for the six distance measures}
{
\small
\label{table:efficiency}
\begin{tabular}{c|c|rrrrrr}
\hline	
Distance      & Algorithms    & bat     & seahorse     & pegasus   &  squid  &  octopus & spider  \\
measures      &               & $n=58$  & $n=60$       & $n=60$    &  $n=92$ & $n=94$   & $n=120$ \\
\hline	
\hline	
Euclidean     &  Algorithm 1  &  0.6    & 0.8          &  0.7      &  4.7       &  4.3       &  16.1   \\
\hline	
WE            &  Algorithm 2  &   1.1   &   1.2        &  1.1      &  10.9      &  15.6      &  66.8   \\
              &  naive        & 278.0   & 359.4        & 351.3     & 1.7{\bf h} & 1.9{\bf h} & 10.6{\bf h}  \\
\hline	
AD            &  Algorithm 3  &  1.7    & 2.0          &  2.2      &  19.2      & 27.5       &  117.6  \\
              &  naive        &  20.6   & 25.0         &  25.0     & 343.7      & 411.9      &  2103.6 \\
\hline	
WAD           &  Algorithm 4  &  2.9    & 3.0          &  4.9      & 62.4       & 77.8      &  340.7  \\
              &  naive        &  345.0  & 432.3        &  429.5    & 2.0{\bf h} & 2.3{\bf h} & 12.3{\bf h}   \\
\hline	
GAD1          &  Algorithm 4  &  2.4    & 2.4          &  5.8      &  63.5      &  209.6     & 785.6    \\
($\gamma=1.2)$&  Algorithm 5  &  2.3    & 2.5          &  4.2      &  29.5      &   45.7     &  152.4   \\
\hline	
GAD1          &  Algorithm 4  &  2.4    & 2.6          &   7.4     &  82.6      &  377.9     &  1580.0  \\
($\gamma=1.4)$&  Algorithm 5  &  2.1    & 2.6          &   3.8     &  27.3      &   37.3     &  143.6  \\
\hline	
GAD1          &  Algorithm 4  &  3.0    & 2.8          &  9.6     & 112.8      &  579.6     &  3241.2  \\
($\gamma=1.6)$&  Algorithm 5  &  2.2    & 2.2          &  3.5     &  27.6      &  33.9      &  152.6  \\
\hline	
GAD2          &  Algorithm 4  &  2.6    & 2.4          &  21.3     & 152.8      &  2158.8    & 1.1{\bf h} \\
($\gamma=1.2)$&  Algorithm 5  &  2.3    & 2.5          &   6.8     &  30.0      &   128.1    &  218.9 \\
\hline	
GAD2          &  Algorithm 4  &  3.3    & 2.7          & 43.3      & 474.7      & 1.8{\bf h} & 2.9{\bf h}  \\
($\gamma=1.4)$&  Algorithm 5  &  2.1    & 2.6          &  7.1      &  35.2      &  159.9     &  228.6  \\
\hline	
GAD2          &  Algorithm 4  &  6.0    & 3.6          & 67.5      & 1711.6     & 2.7{\bf h} & 8.7{\bf h}  \\
($\gamma=1.6)$&  Algorithm 5  &  2.1    & 2.3          &  6.4      &  49.4      &  214.5     &  382.6  \\
\hline	
\end{tabular}
%\vspace{-5mm}
}
\end{center}
\end{table}

\section{Discussions} \label{sec:6}

%We discuss the strengths and weaknesses of the GAD distances from the viewpoint of the Escherization problem. 

One of the main contributions of this paper is to introduce the GAD (GAD1 and GAD2) distances into the exhaustive search of the templates. A common concept of the GAD distances is to focus on the similarity of local structures between the goal and tile polygons. One may think that the set $E_c$ should be defined as a set of edges between all possible pairs of points of the goal polygon to consider not only the local structural similarity but also the global structural similarity (e.g. the overall positional relationship between all point). In our preliminary experiments, however, the use of such a distance measure resulted in almost the same tile shapes as when the Euclidean distance was used. In this sense, the Euclidean distance is most suitable to focus on the global structural similarity. It is expected that the balance between the global and local structural similarities can be adjusted by the parameter $\gamma$. As described in the experimental results section, however, the most satisfactory tile shapes were obtained when $\gamma$ was set from 1.2 to 1.6 (see the set $E_c$ shown in Fig. \ref{fig:result1}). This suggests that maintaining the local structures is a necessary condition for obtaining a satisfactory tile shape and the GAD distances should focus only on the local structural similarity.

In fact, the top ten tile shapes obtained with the GAD distances were rich in variety and some of them were very unsatisfactory because the global structure was sometimes seriously distorted. However, maintaining the local structure indirectly leads to maintaining the global structure, and the global structure was maintained to some extent in some of the top ten tile shapes obtained with the GAD distances. In such a case, we can obtain a very satisfactory tile shape that preserves both local and global structures at a high level. For this reason, the tile shape with the smallest distance value does not necessarily the one that is intuitively most similar to a given goal polygon (this is also the case when the Euclidean distance is used). This is not a big problem because in our observation the most satisfactory tile shape was usually included in the top ten tile shapes, and we can easily find the intuitively best one from them. However, it is desirable to develop a distance function that better matches the human sense. 
%(typical examples are found in the result for goal polygon pegasus in the supplementary file)

A possible drawback of the GAD distances (with appropriate parameter values for $\gamma$) is to focus on the similarity of the local structure only near the boundary surface. Therefore, a tile shape can be highly evaluated as long as the local structure near the boundary is well preserved even if the inside of the goal shape is excessively distorted to form the tile shape. Several techniques to address this problem have been proposed in the research field of the 3D (or 2D) surface modeling and its editing \cite{zhou2005large,igarashi2005rigid,weng20062d}. A basic idea is to represent an object as a mesh (polygon in the context of the Escherization problem) consisting of vertices arranged not only on the surface but also inside the object. To measure the naturalness of the deformation of the object, deformation energy functions (distance functions in the context of the Escherization problem) are defined as the squared sum of the differences in the relative positional relationship of adjacent points between the original and deformed objects. The same idea would be applicable to the distance function for the Escherization problem and we leave this for future work.

\section{Conclusion} \label{sec:7}

In this study, several distance functions suitable for Koizumi and Sugihara's formulation of the Escherization problem have been proposed to search for intuitively more satisfactory tile shapes than those obtained with the Euclidean (or Procrustes) distance. The weighted Euclidean distance emphasizes the similarity with important parts of the goal polygon specified by the user. The AD distance focuses on the similarity of the relative positional relationship of adjacent points, and its weighted variant is also proposed. The GAD distance is a generalization of the AD distance, where the similarity of the relative positional relationship of close points is considered without restricting adjacent points, and two types of the GAD distance (GAD1 and GAD2) are designed. The proposed distance functions are incorporated into the exhaustive search of the templates. It was demonstrated that the use of the proposed distance functions, especially the GAD1 or GAD2 distance, created more satisfactory tile shapes than the ones obtained with the Euclidean distance for the six goal polygons consisting of 58 to 120 points.   

The efficient exhaustive search algorithms developed for the weighted Euclidean, the AD distance, and the weighted AD distance were executed in a reasonable computation time (at most 340 s on a 120-point goal polygon). However, for the GAD distance, the efficient exhaustive search algorithm was time-consuming for relatively large goal polygons. Therefore, an incomplete search algorithm has also been developed for the GAD distance to obtain results in a reasonable computation time, which was at most 382 s on a 120-point goal polygon. It was confirmed that the incomplete search algorithm for the GAD distance rarely overlooked the top 10 solutions obtained by the exhaustive search algorithm. 

\section*{Acknowledgement}
This work was supported by JSPS KAKENHI Grant Number 20K11695.

\bibliographystyle{spmpsci} 
\bibliography{NAGATA}

\section*{APPENDIX}
\appendix

\section{Formulation of the Proposed methods with the Procrustes distance}

We have described Koizumi and Sugihara's original formulation of the Escherization problem \cite{koizumi2011maximum}, the exhaustive search of the templates \cite{nagata2019}, and the proposed distance functions based on the case where the Euclidean distance (Eq.~\ref{eq:euclid}) is available. However, we need to use the Procrustes distance (Eq.~\ref{eq:procrustes}) for some isohedral types (see Section \ref{sec:2_parameterization}). This appendix describes how these are modified when the Procrustes distance must be used. 

In \cite{nagata2019}, a simplified version of the Procrustes distance (see footnote 1) was used instead of the original one because it must be used to implement the third technique of Algorithm \ref{alg:E} (see Section \ref{sec:4_Euclidean}). The simplified Procrustes distance $d_P(U,W)$ can be expressed as a function of the vector ${\bm u}$ as follows: 
\begin{equation}
\label{eq:procrustes_simple}
d_P^2(U,W) = \min_{\theta} {\left\| R(\theta) U - W \right\|}^2
= {\bm{w}}^{\top} \bm{w} + {\bm{u}}^{\top} \bm{u} - 2 \sqrt{{\bm{u}}^{\top} V {\bm{u}}}, 
\end{equation}
where $V$ is the $2n \times 2n$ symmetric matrix defined by 
\begin{equation}
\label{eq:V}
V = \bm{w} \bm{w}^{\top} + \bm{w_c} \bm{w_c}^{\top}
\end{equation}
where $\bm{w_c}$ is the $2n$-dimensional vector defined as $\bm{w_c} = {(y^w_1, \dots, y^w_n, -x^w_1, \dots, -x^w_n)}^{\top}$. Note that Koizumi and Sugihara \cite{koizumi2011maximum} showed that the original Procrustes distance (Eq.~\ref{eq:procrustes}) can be transformed into $1 - \frac{1}{{\bm w}^{\top}{\bm w}} \frac{{\bm{u}}^{\top} V \bm{u}} {{\bm{u}}^{\top}\bm{u}}$, and the last expression of Eq.~\ref{eq:procrustes_simple} is obtained in a similar manner. Recall that we need to consider $n$ different numbering schemes of the goal polygon $W$, and $W_j \ (j = 1, 2, \dots, n)$ are defined. Therefore, we need to define $V_j$ individually for the goal polygons $W_j \ (j=1, 2, \dots, n)$. 

First, we explain how the optimization problem (\ref{eq:formulation_ikj_Euclid}) and the optimal value $eval^E_{ikj}$ (Eq.~\ref{eq:eval_ikj_Euclid}) are modified when the simplified Procrustes distance is used. From Eqs.~\ref{eq:u=Bxi} and \ref{eq:procrustes_simple}, the exhaustive search of the templates combined with the simplified Procrustes distance can be formulated as the following unconstrained optimization problem:
\begin{equation}
\label{eq:optimization_procrustes}
\argmin_{\bm{\xi}} \ {\bm{w}}^{\top} \bm{w} + {\bm{\xi}}^{\top} \bm{\xi} - 2 \sqrt{{\bm{\xi}}^{\top} {B_{ik}}^{\top} V_j B_{ik} {\bm{\xi}}},
\end{equation}
for all combinations of  $i \in I$, $k \in K_i$, and $j \in J$. For each triplet $(i,k,j)$, the optimal solution must satisfy the equation $\displaystyle {\bm{\xi}} - \frac{{B_{ik}}^{\top} V_j B_{ik} {\bm{\xi}}}{\sqrt{{\bm{\xi}}^{\top} {B_{ik}}^{\top} V_j B_{ik} {\bm{\xi}}}} = {\bm 0}$, which is obtained by setting the term of Eq.~\ref{eq:optimization_procrustes} differentiated by ${\bm{\xi}}$ to the zero vector. We can see that this equation is equivalent to 
\begin{eqnarray}
{B_{ik}}^{\top} V_j B_{ik} {\bm{\xi}} &=& \lambda {\bm{\xi}} \\
\left\| {\bm{\xi}} \right\|^2 &=& \lambda, 
\end{eqnarray}
where $\lambda$ is a scalar variable. Therefore, the solutions of this equation are the eigenvectors of ${B_{ik}}^{\top} V_j B_{ik}$, each of whose length is the square root of the corresponding eigenvalue. Let ${\bm{\xi}^*}$ and $\lambda$ be such an eigenvector and the corresponding eigenvalue. Then, we have ${\bm{\xi}^*}^{\top} \bm{\xi^*} - 2 \sqrt{{\bm{\xi}^*}^{\top} {B_{ik}}^{\top} V_j B_{ik} {\bm{\xi}^*}} = - \lambda$. Consequently, for each triplet $(i,k,j)$, the optimal value of the optimization problem (\ref{eq:optimization_procrustes}) is given by 
\begin{equation}
\label{eq:eval_ikj_procrustes}
eval^P_{ikj} = {\bm{w}}^{\top} \bm{w} - \lambda_{ikj}, 
\end{equation}
where $\lambda_{ikj}$ is the maximum eigenvalue of ${B_{ik}}^{\top} V_j B_{ik}$. The optimal solution ${\bm{\xi^*_{ikj}}}$ is then given by the corresponding eigenvector whose length is $\sqrt{\lambda_{ikj}}$. The maximum eigenvalue $\lambda_{ikj}$ and the corresponding eigenvector $\bm{\xi_{ikj}^*}$ can be computed in $O(n)$ time (the second technique of Algorithm \ref{alg:E}) \cite{nagata2019}.

Subsequently, we consider how the optimization problem (\ref{eq:formulation_general}) and its optimal value ${eval_{ikj}^G}$ (Eq.~\ref{eq:eval_ikj_general}) are modified when the simplified Procrustes distance is extended in the same way as the Euclidean distance is extended to the generalized quadratic form given by Eq.~\ref{eq:mahalanobis}. As ${\|X\|}^2 = \mathrm{Tr}(X X^{\top})$ for any matrix $X$, the simplified Procrustes distance is represented by 
\begin{equation}
\label{eq:tr_procrustes_ikj}
d_P^2(U,W) = \min_{\theta} \mathrm{Tr} \left(R(\theta) U - W \right) \left(R(\theta) U - W \right)^{\top}. 
\end{equation}
Then, the generalized version of the simplified Procrustes distance is expressed as follows: 
\begin{equation}
\label{eq:mahalanobis_procrustes_ikj}
d_{GP}^2(U,W) = \min_{\theta} \mathrm{Tr} \left(R(\theta) U - W \right) K \left(R(\theta) U - W \right)^{\top}. 
\end{equation}
As the matrix $K$ is a symmetric positive-definite matrix (of size $n$), there exists a square matrix $A$ such that $K = A^{\top} A$. Therefore, we have $d_{GP}^2(U,W) = \min_{\theta} \mathrm{Tr} \left(R(\theta) UA^{\top} - WA^{\top} \right)$ $\left(R(\theta) UA^{\top} - WA^{\top} \right)^{\top}$. Therefore, the right-hand side Eq.~\ref{eq:mahalanobis_procrustes_ikj} is obtained from the right-hand side Eq.~\ref{eq:tr_procrustes_ikj} simply by replacing $U$ and $W$ with $UA^{ \top}$ and $WA^{\top}$, respectively. From the definition of ${\bm u}, {\bm w}$, and $V$ (see Section \ref{sec:2_parameterization}), this is equivalent to replacing these values in Eq.~\ref{eq:procrustes_simple} as follows: ${\bm u} \rightarrow G^{\frac{1}{2}} {\bm u}$, ${\bm w} \rightarrow G^{\frac{1}{2}} {\bm w}$, and $V \rightarrow {G^{\frac{1}{2}}}^{\top} V G^{\frac{1}{2}}$, where $G^{\frac{1}{2}} = \begin{pmatrix}
A & O \\
O & A \\
\end{pmatrix}
$.
For example, let $U$ be represented as 
$U = 
\begin{pmatrix}
{\bm{u_x}}^{\top} \\ 
\bm{u_y}^{\top}
\end{pmatrix}  
$ 
, where $\bm{u_x} = {(x_1, \dots, x_n)}^{\top}$ and $\bm{u_y} = {(y_1, \dots, y_n)}^{\top}$. The replacement $U \rightarrow U A^{\top}$ is equivalent to 
$
\begin{pmatrix}
{\bm{u_x}}^{\top} \\ 
\bm{u_y}^{\top}
\end{pmatrix}  
\rightarrow 
\begin{pmatrix}
{\bm{u_x}}^{\top} A^{\top} \\ 
\bm{u_y}^{\top} A^{\top}
\end{pmatrix}  
$, and then 
$
{\bm u} = 
\begin{pmatrix}
{\bm{u_x}} \\ 
\bm{u_y}
\end{pmatrix}  
\rightarrow 
\begin{pmatrix}
A {\bm{u_x}} \\ 
A \bm{u_y}
\end{pmatrix}
= 
\begin{pmatrix}
A & O \\
O & A \\
\end{pmatrix}
\bm{u}
$.
Considering that ${G^{\frac{1}{2}}}^{\top} G^{\frac{1}{2}} = \begin{pmatrix}
K & O \\
O & K \\
\end{pmatrix}
= G$ (symmetric matrix),
$d_{GP}^2(U,W)$ is expressed as a function of ${\bm u}$ as follows:
\begin{equation}
\label{eq:mahalanobis_procrustes_simple}
d_{GP}^2(U,W) = {\bm{w}}^{\top} G \bm{w} + {\bm{u}}^{\top} G \bm{u} - 2 \sqrt{{\bm{u}}^{\top} G V G {\bm{u}}}. 
\end{equation}

From Eqs.~\ref{eq:u=Bxi} and \ref{eq:mahalanobis_procrustes_simple}, the exhaustive search of the templates combined with the generalized version of the simplified Procrustes distance can be formulated as the following unconstrained optimization problem:
\begin{equation}
\label{eq:optimization_mahalanobis_procrustes}
\argmin_{\bm{\xi}} \ {\bm{w_j}}^{\top} G_j \bm{w_j} + {\bm{\xi}}^{\top} {B_{ik}}^{\top} G_j B_{ik} \bm{\xi} - 2 \sqrt{{\bm{\xi}}^{\top} {B_{ik}}^{\top} {G_j} V_j G_j B_{ik} {\bm{\xi}}}, 
\end{equation}
for all combinations of  $i \in I$, $k \in K_i$, and $j \in J$. With the same calculation as in the case of the simplified Procrustes distance, we can observe that, for each triplet $(i,k,j)$, we need to solve the generalized eigenvalue problem 
\begin{equation}
\label{eq:generalized_eigenvalue_problem}
{B_{ik}}^{\top} {G_j} V_j G_j B_{ik} {\bm{\xi}} = \lambda {B_{ik}}^{\top} G_j B_{ik} {\bm{\xi}}
\end{equation}
to solve the optimization problem (\ref{eq:optimization_mahalanobis_procrustes}). The optimal value is then given by 
\begin{equation}
\label{eq:eval_ikj_general_procrustes}
eval_{ikj}^{GP} = {\bm{w_j}}^{\top} G_j \bm{w_j} - \lambda_{ikj},  
\end{equation}
where $\lambda_{ikj}$ is the maximum eigenvalue of this generalized eigenvalue problem. The optimal solution $\bm{\xi_{ikj}^*}$ is then given by the corresponding eigenvector whose length is determined to satisfy ${\bm{\xi_{ikj}^*}}^{\top} {B_{ik}}^{\top} G_j B_{ik} \bm{\xi_{ikj}^*} = \lambda_{ikj}$. 

In general, for a symmetric matrix $A$ and a symmetric positive-definite matrix $B$, the generalized eigenvalue problem $A {\bm x} = \lambda B {\bm x}$ can be converted into the eigenvalue problem \\ $(L^{-1} A (L^{-1})^{\top}) {\bm y} = \lambda {\bm y}$, where ${\bm y} = L^{\top} {\bm x}$ and $L$ is a lower triangular matrix such that $B = L L^{\top}$ (the Cholesky decomposition). This technique was used to solve the generalized eigenvalue problem (\ref{eq:generalized_eigenvalue_problem}), but it takes $O(n^3)$ time to perform the Cholesky decomposition of the matrix ${B_{ik}}^{\top}  G_j B_{ik}$. Once Eq.~\ref{eq:generalized_eigenvalue_problem} is converted into the eigenvalue problem, the maximum eigenvalue and the corresponding eigenvector can be computed in $O(n)$ time as in the case of the simplified Procrustes distance. 

\end{document}